\documentclass[final]{mathincs}

\usepackage{amsmath,amssymb}
\usepackage{mathtools}
\usepackage[linesnumbered,ruled,vlined,noend]{algorithm2e}
\usepackage{graphicx}

\usepackage{booktabs}
\usepackage{subcaption}
\usepackage{url}

\usepackage{tikz}
\usetikzlibrary{arrows}

\newtheorem{theorem}{Theorem}[section]
\newtheorem{lemma}[theorem]{Lemma}
\newtheorem{corollary}[theorem]{Corollary}
\newtheorem{proposition}[theorem]{Proposition}

\theoremstyle{definition}
\newtheorem{definition}[theorem]{Definition}
\newtheorem{example}[theorem]{Example}
\newtheorem{remark}[theorem]{Remark}

\DeclareMathOperator{\NR}{NR}

\DeclareMathOperator{\Supp}{Supp}

\newcommand\LT{\mathop{\rm LT}\nolimits}

\newcommand\tsum{\textstyle\sum\limits}
\newcommand\rangleF{\rangle_{\mathbb{F}_2}}

\let\phi=\varphi
\let\rho=\varrho

\newcommand\ZZ{{\mathbb{Z}}}

\newcommand\FF{{\mathbb{F}}}
\newcommand\NN{{\mathbb{N}}}
\newcommand\BB{{\mathbb{B}}}
\newcommand\LL{{\mathbb{L}}}
\newcommand\II{{\mathbb{I}}}

\DeclareMathOperator{\True}{\mathtt{true}}
\DeclareMathOperator{\False}{\mathtt{false}}

\DeclareMathOperator{\calS}{\mathcal{S}}
\DeclareMathOperator{\calZ}{\mathcal{Z}}

\DeclareMathOperator{\calO}{\mathcal{O}}

\renewcommand{\iff}{\leftrightarrow}
\DeclareMathOperator*{\LXOR}{\bigoplus}
\DeclareMathOperator*{\LAND}{\bigwedge}
\DeclareMathOperator*{\LOR}{\bigvee}

\newcommand{\GCP}{\ensuremath{\mathtt{GGCP}}\xspace}
\newcommand{\PP}{\ensuremath{\mathtt{PP}}\xspace}
\newcommand{\eGCP}{\ensuremath{\mathtt{crGGCP}}\xspace}
\newcommand{\SCC}{\ensuremath{\mathtt{SCC}}\xspace}
\newcommand{\FLS}{\ensuremath{\mathtt{tFLS}}\xspace}
\newcommand{\UNSAT}{\ensuremath{\mathtt{UNSAT}}\xspace}
\newcommand{\SOLVER}{\ensuremath{\mathtt{G{\_\kern0.08em}2XNF{\_\kern0.12em}DPLL}}\xspace}
\newcommand{\txnfsolver}{\ensuremath{\mathtt{2\text{-}Xornado}}\xspace}

\newcommand{\cms}{\ensuremath{\mathtt{CryptoMiniSat}}\xspace}
\newcommand{\sbva}{\ensuremath{\mathtt{SBVA{\text{-}}CaDiCaL}}\xspace}
\newcommand{\cadical}{\ensuremath{\mathtt{CaDiCaL}}\xspace}
\newcommand{\xnfsat}{\ensuremath{\mathtt{xnfSAT}}\xspace}
\newcommand{\xnfbf}{\ensuremath{\mathtt{xnf{\_\kern0.1em}bf}}\xspace}
\newcommand{\PolyBoRi}{\ensuremath{\mathtt{PolyBoRi}}\xspace}

\newcommand{\Bosphorus}{\ensuremath{\mathtt{Bosphorus}}\xspace}
\newcommand{\OptiMathSAT}{\ensuremath{\mathtt{OptiMathSAT}}\xspace}
\newcommand{\MaxHS}{\ensuremath{\mathtt{MaxHS}}\xspace}

\newcommand{\XnfToTXnf}{\ensuremath{\mathtt{XNFto2XNF}}\xspace}
\newcommand{\AnfToTXnf}{\ensuremath{\mathtt{ANFto2XNF}}\xspace}
\newcommand{\QAnfToTXnf}{\ensuremath{\mathtt{QANFto2XNF}}\xspace}

\newcommand{\Ascon}{\ensuremath{\mathtt{Ascon}}\xspace}
\newcommand{\Asconnb}{\ensuremath{\mathtt{Ascon\text{-}128}}\xspace}

\newcommand{\SRES}{\ensuremath{\mathtt{s\text{-}Res}}\xspace}

\newcommand{\cpp}{\ensuremath{\mathtt{C\text{++}}}\xspace}
\newcommand{\lex}{\ensuremath{\mathtt{lex}}\xspace}

\title[XNF SAT-Solving]{SAT Solving Using XOR-OR-AND Normal Forms}

\author{Bernhard Andraschko}
\address[Bernhard Andraschko]{Fakult{\"a}t f{\"u}r Informatik und Mathematik\\
Universit\"{a}t Passau, D-94030 Passau, Germany}
\email{bernhard.andraschko@uni-passau.de}

\author{Julian Danner}
\address[Julian Danner]{Fakult{\"a}t f{\"u}r Informatik und Mathematik\\
Universit\"{a}t Passau, D-94030 Passau, Germany}
\email{julian.danner@uni-passau.de}

\author{Martin Kreuzer}
\address[Martin Kreuzer]{Fakult{\"a}t f{\"u}r Informatik und Mathematik\\
Universit\"{a}t Passau, D-94030 Passau, Germany}
\email{martin.kreuzer@uni-passau.de}

\date{24.10.2024}

\begin{document}

\begin{abstract}
This paper introduces the XOR-OR-AND normal form (XNF) for logical formulas. It
is a generalization of the well-known Conjunctive Normal Form (CNF) where
literals are replaced by XORs of literals. As a first theoretic result, we show
that every CNF formula is equisatisfiable to a formula in 2-XNF, i.e., a formula in
XNF where each clause involves at most two XORs of literals.
Subsequently, we present an algorithm which converts Boolean polynomials efficiently 
from their Algebraic Normal Form (ANF) to formulas in 2-XNF. Experiments with the
cipher ASCON-128 show that cryptographic problems, which by design are based
strongly on XOR-operations, can be represented using far fewer variables and
clauses in 2-XNF than in CNF. 
In order to take advantage of this compact representation, new SAT solvers based
on input formulas in 2-XNF need to be designed. By taking inspiration from
graph-based 2-CNF SAT solving, we devise a new DPLL-based SAT solver for
formulas in 2-XNF. Among others, we present advanced pre- and in-processing
techniques.
Finally, we give timings for random 2-XNF instances and instances related to key
recovery attacks on round reduced ASCON-128, where our solver outperforms state-of-the-art
alternative solving approaches.
\end{abstract}


\maketitle

%%%%%%%%%%%%%%%%%%%%%%%%%%
% Section 1: Introducton
%%%%%%%%%%%%%%%%%%%%%%%%%%

\section{Introduction}

SAT solvers are programs which decide the Boolean Satisfiability Problem
for propositional logic formulas. In the last decades there has been 
a substantial effort to improve their performance, and they have grown into
versatile tools for tackling computational problems in various domains such as
automatic theorem proving, graph theory, hardware verification, artificial
intelligence, and cryptanalysis.

Especially problems from the latter domain have been shown to be hard for
conventional SAT solvers that take a conjunctive normal form (CNF) as
input. Although many new attacks on cryptosystems and other cryptographic protocols
have been designed based on the idea of encoding the computational problem as an
instance for CNF-based SAT solvers (see for instance~\cite{DKMNPW, HE, JK,
LNV, LZLW, MZ}), many problems are still out of range (e.g., see~\cite{DK, DKMNPW}).
This can be mainly attributed to the fact that cryptographic primitives are
often built using \textit{exclusive disjunctions (XORs)} of variables which lead
to an exponential blow-up when encoded in~CNF. 

To speed up the performance of SAT solvers for such instances, one can either try to 
modify the problem encodings such that they lead to smaller instances (see~\cite{CSCM,HK2,JK}), or one can try to improve the solving strategy altogether. 
For example, the latter approach has been pursued by attempts to integrate support for XOR 
constraints on the input variables (see \cite{HJ,NLFHB,SNC,TID})
or by combining logical SAT solving with algebraic reasoning (see~\cite{CSCM,Hor}). 
While purely algebraic solving techniques (as developed for instance in \cite{BD,CKPS,CSSV}) 
have had some success in cryptanalysis (e.g., see \cite{ACGKLRS,CD,DK}), a very
promising line of research seems to be to combine logic and algebraic solving paradigms. 

One such attempt was initiated in~\cite{HK1} and refined in~\cite{Hor}, where a new proof 
system called \SRES\ was introduced. Its input are products of linear Boolean
polynomials, or, in the language of logic,  disjunctions of XORs of literals. 
Thus the \SRES\ proof system is innately suitable for dealing with cryptanalytic instances,
as these tend to be rich in XOR connectives.
The core inference rule of~\SRES\ is called \textit{$s$-resolution}. It is both a generalization of the classical 
resolution rule of propositional logic and of Buchberger's S-polynomials which form the
basis of Gr\"obner basis computations (see~\cite{KR1}).
In~\cite{Hor} and~\cite{HK1} initial DPLL-based refutation methods utilizing 
$s$-resolvents were introduced. As of today, no highly efficient implementation of 
these algorithms exists, and procedures to use \SRES\ for finding satisfying
assignments are lacking as well.
\smallskip

In the first part of this paper we strive to develop efficient methods for converting propositional 
logic formulas to suitable inputs for \SRES\ type proof systems. After recalling some basic
definitions and properties of the ring of Boolean polynomials in Section~\ref{sec:basics},
we introduce and study a new XOR-based normal form in Section~\ref{sec:XNF}.
More precisely, the new normal form is called the \textit{XOR-OR-AND normal form (XNF)} for propositional logic
formulas (see Definition~\ref{def:xnf}). It generalizes the CNF by replacing the literals with XORs of
literals which we simply call \textit{linerals}. From an algebraic perspective, 
a lineral corresponds to a linear Boolean polynomial, and a disjunction of linerals 
corresponds to a product of linear Boolean polynomials. Using this identification, one sees 
that formulas in XNF occur naturally in the proof system~\SRES. Since XNF generalizes CNF, it is
clear that every propositional logic formula is equivalent to one in XNF. 
%However, the true task is to find an XNF representation which is as compact as possible. 

While conversions from the algebraic normal form (ANF) of a Boolean polynomial to the CNF of 
the corresponding propositional logic formula and back have been studied carefully (see for instance~\cite{CSCM,HK2,JK}), conversions to systems offering some native support for XOR 
have been introduced only sparingly and elaborated much less systematically (see~\cite{LJN,NLFHB,SNC}).
 
It is well-known that one can introduce new variables and convert every Boolean polynomial system
to one involving only polynomials of degree at most two. Here we show that, in fact, every
XNF formula is equisatisfiable to one in 2-XNF, i.e., to an instance of XNF where each clause involves
at most two linerals (see Proposition~\ref{prop:2xnf}). Algebraically speaking, systems of quadratic 
Boolean polynomial equations can be transformed to systems consisting of products of at most two linear polynomials.
Furthermore, we try to optimize this transformation by introducing as few additional variables
as possible (see Propositions~\ref{prop:alg:polyTo2Xnf} and~\ref{prop:alg:quadPolyTo2Xnf}). 
To illustrate the potential of the conversion to 2-XNF, we 
apply it to instances related to algebraic attacks on the cipher \Asconnb\
(see~\cite{DEMS}) which was recently selected for standardization by NIST for lightweight cryptography.
We get 2-XNF representations which are substantially more compact than state-of-the-art representations in CNF
(see Example~\ref{exa:ascon_sbox}). 
\smallskip

In the second part of the paper we make use of this 2-XNF representation and take the first few
steps towards translating the foundations of efficient CNF-based SAT solving to 
XNF-based SAT solving. In particular, using ideas based on efficient 2-CNF solvers
and CNF pre-processing (see~\cite{APT,HMB}), we develop a graph based 2-XNF solver. 
To start with, we define an \textit{implication graph structure (IGS)} $(L,V,E)$ for a given formula~$F$
which consists of a set~$L$ of linear Boolean polynomials known to be in the
ideal $I_F$, which is the algebraic representation of~$F$, and a directed graph $(V,E)$
whose edges $(f,g)$ mean that $f\in I_F$ implies $g\in I_F$ (see Definition~\ref{def:igs}
and Remark~\ref{rem:motivation}). Our solving algorithm then starts with a trivial
IGS for~$F$ (see Remark~\ref{rem:trivialIGS}) and simplifies it using a suitable ordering
on the IGSs (see Definition~\ref{def:IGSordering}). 
Then we gradually improve the IGS by propagation, in-processing and
guessing until we arrive at an implication graph structure with an empty graph, i.e.,
a case where the corresponding ideal is generated by linear polynomials.
Given that the guesses were correct, a satisfying assignment for~$F$ can be
deduced immediately from a solution of the corresponding system of linear
equations. The improvement of an IGS is measured in terms of
the size of the linear part~$L$ and in the size of the graph $(V,E)$. 

Propagation is achieved using a generalization of the classical Boolean constraint propagation
which we call \textit{Gau{\ss}ian Constraint Propagation} (see Proposition~\ref{prop:alg:GCP}).
Two pre-processing methods are examined which yield new linear forms or new edges for the IGS
(see Proposition~\ref{prop:preprocessing}). Unfortunately, they are too expensive to be
executed repeatedly during the main solving procedure. For such in-processing methods,
we provide two more efficient suggestions. Firstly, using the calculation of strongly connected
components of~$(V,E)$, we are able to reach an acyclic graph (see Proposition~\ref{prop:alg:eGCP}).
Secondly, we introduce the notion of \textit{failed linerals} (see Definition~\ref{def:failed})
and apply them in order to learn new linear polynomials in~$I_F$ (see Proposition~\ref{prop:alg:FLS}).

Several heuristics for producing good \textit{decisions} for an IGS, i.e., for making good 
guesses (see Definition~\ref{def:decision}) are discussed next. Moreover, we offer
some suggestions how to implement these heuristics efficiently (see Remarks~\ref{rem:decheu}
and~\ref{rem:decheu_impl}). Finally, we combine everything and present our new graph-based
2-XNF solver (see Proposition~\ref{prop:alg:solver}) together with suggestions how to
implement it well using suitable data structures (see Remark~\ref{rem:DataStructures}).

The last section contains the results of some experiments and comparisons to
established CNF-based SAT solvers, especially ones that offer some support for XOR constraints.
Usually, they allow separate XOR constraints on variables in addition to a CNF, a type
of input that is known as CNF-XOR. One of the first solvers for such formulas was described in~\cite{SM,SNC}
and is implemented in \cms. Another one is \xnfsat\ (see~\cite{NLFHB}) which uses stochastic local search methods for  CNF-XOR inputs. 
The solver \Bosphorus\ allows ANF and CNF inputs, but no CNF-XOR inputs (see~\cite{CSCM}). Moreover, we compare our method with the winner of the 2023 SAT competition \sbva (see~\cite{HGH}) which admits CNF formulas.

In our experiments we compare the new 2-XNF solver using the three decision heuristics 
explained in Section~\ref{sec:solving} to the CNF-XOR solvers \cms\ and \xnfsat, to the
algebraic solver \PolyBoRi\ (see~\cite{BD}), to a brute force XNF solver \xnfbf, and to the CNF solver \sbva.
We created timings for two types of inputs. The first type are random 2-XNF instances. It turns out that
our graph based 2-XNF solver involving merely some basic DPLL techniques outperforms
state-of-the-art solving approaches on small random instances, both for satisfiable 
and unsatisfiable cases (see Figures~\ref{fig:cactus_sat} and~\ref{fig:cactus_unsat}).

The second type of experiments was to try the solvers on key recovery attacks for
round reduced versions of the \Asconnb\ cryptosystem. This lightweight cipher was
recently selected for standardization by NIST. Again it turns out that, even with our very simple
decision heuristics, the graph based 2-XNF solver performs surprisingly well 
(see Figure~\ref{fig:bench_rr_ascon}). Here it may be worthwhile to note that 
some of these round reduced key recovery attacks can be solved already in the pre-processing
phase. Altogether, one main advantage of XNF solving
is that the more compact problem representations require fewer decisions, and this leads
to meaningful speed-ups. Finally, let us point out that the desirable extension of
XNF solving to include CDCL techniques is not straightforward and will require 
non-trivial new tools.
\smallskip

Due to its simpler description, we chiefly use the algebraic point of view when
we work with formulas in XNF, i.e., we regard them as products of linear Boolean polynomials.
Unless explicitly noted otherwise, we use the definitions and notation
introduced in~\cite{HK2} and~\cite{KR1}. The algorithms of
Section~\ref{sec:XNF} were implemented by B.~Andraschko and the solving
methods of Section~\ref{sec:solving} by J.~Danner. All source code is
available at \url{https://github.com/j-danner/2xnf_sat_solving}.

%%%%%%%%%%%%%%%%%%%%%%%%%%%%%%%%%%%%%%%%%%%%%%%%%%%%%%%%%%%%
%
% Section 2: The Ring of Boolean Polynomials
%
%%%%%%%%%%%%%%%%%%%%%%%%%%%%%%%%%%%%%%%%%%%%%%%%%%%%%%%%%%%%

\section{The Ring of Boolean Polynomials}%
\label{sec:basics}

Throughout this paper we let $\FF_2$ be the field with two elements, $n \in
\NN_+$, and $P = \FF_2[X_1,\dots,X_n]$ the polynomial
ring over~$\FF_2$ in the indeterminates $X_1,\dots,X_n$. Recall that
the \textit{ring of Boolean polynomials} is
$$
\BB_n \;=\; P/ \langle X_1^2-X_1,\dots,X_n^2-X_n \rangle 
$$
where $\II_n = \langle X_1^2 - X_1, \dots, X_n^2 - X_n \rangle$ is also called the
\textit{field ideal} in~$P$.
Whenever additional indeterminates are required, we write 
$$
\BB_{n,m} \;=\; \FF_2[X_1,\dots,X_n, Y_1,\dots,Y_m]/\II_{n,m}
$$
where $\II_{n,m} = \langle X_1^2-X_1, \dots, X_n^2-X_n, Y_1^2-Y_1,\dots, Y_m^2-Y_m\rangle $.
For $i\in\{1,\dots,n\}$ and $j\in\{1,\dots,m\}$, we denote the residue class of~$X_i$ in~$\BB_n$ and $\BB_{n,m}$
by~$x_i$ and the residue class of~$Y_j$ in~$\BB_{n,m}$ by~$y_j$. These residue classes will be called
the \textit{indeterminates} of~$\BB_n$ and~$\BB_{n,m}$, respectively, and the elements of these rings are
called \textit{Boolean polynomials}.

Every Boolean polynomial $f\in \BB_n$ can be uniquely written 
as a sum of distinct square-free terms, where a \textit{term} is a product of distinct 
residue classes~$x_i$. This is known as the \textit{algebraic normal form (ANF)}
of~$f$. (See for instance \cite[Sec.~2.1]{Hor} or~\cite{Bri} for an in-depth
study of ANFs.) Altogether, we have $\BB_n = \FF_2[x_1,\dots,x_n]$ as an $\FF_2$-algebra 
and $\dim_{\FF_2}(\BB_n) = 2^n$.

Given $f\in \BB_n\setminus \{0\}$ in ANF, replacing each $x_i$ by~$X_i$ yields its
\textit{canonical representative} $F\in P$. Then the \textit{support} of~$f$
is $\Supp(f) = \{t+\II_n \mid t \in \Supp(F)\}$ and 
the \textit{degree} of~$f$ is given by 
$$
\deg(f) \;=\; \min\{\deg(F) \mid F\in P,\; f = F+\II_n\}.
$$
The $\FF_2$-linear span of all Boolean polynomials of degree $\le 1$ plays a major role subsequently.
It is denoted by 
$$
\LL_n \;=\; \langle 1,x_1,\dots,x_n \rangleF \;=\; \FF_2 \oplus \FF_2\, x_1 \oplus \cdots \oplus \FF_2\, x_n 
$$
and called the vector space of \textit{linear Boolean polynomials}. (Note that $\LL_n$ includes the element~1.)

In Section~\ref{sec:solving} we also need division with remainders for Boolean polynomial
rings. Let $\sigma$ be a term ordering on~$P$, and let $f, g_1,\dots,g_s\in\BB_n$ be 
Boolean polynomials in ANF. Let $F, G_1,\dots,G_s \in P$ be the canonical representatives of
$f,g_1,\dots,g_s$, respectively. Then the \textit{normal remainder} of~$f$ under the division
by $G=(g_1,\dots,g_s)$ is defined by 
$$
\NR_{\sigma,G}(f) \;=\; \NR_{\sigma,(G_1,\dots,G_s)}(F) + \II_n .
$$
Moreover, we denote the ordering induced by~$\sigma$ on the terms in~$\BB_n$ by~$\sigma$ again
and call $\LT_\sigma(f) = \LT_\sigma(F)+\II_n$ the \textit{leading term} of~$f$ with respect to~$\sigma$.
For the definitions and an explanation of these concepts in~$P$ see~\cite{KR1}, Chapter~I, and
for more details about orderings on Boolean polynomial rings see~\cite{Bri}.

%%%%%%%%%%%%%%%%%%%%%%%%%%%%%%%%%%%%%%%%%%%%%%%%%%%%%%%%%%%%%%%%%%%
%
% Section 3: Logical Representations of Boolean Polynomials
%
%%%%%%%%%%%%%%%%%%%%%%%%%%%%%%%%%%%%%%%%%%%%%%%%%%%%%%%%%%%%%%%%%%%

\section{Logical Representations of Boolean Polynomials}%
\label{sec:XNF}

In the following we let $\BB_n = \FF_2[x_1,\dots,x_n]$ be the ring of Boolean polynomials.
Our goal is to connect Boolean polynomials to propositional logic formulas.
We assume that the readers are familiar with the syntax of propositional logic 
and identify $\True \equiv 1$ as well as $\False \equiv 0$.

\begin{definition} Let $S$ be a subset of~$\BB_n$, and let $F$ be a propositional logic formula
in the logical variables $X_1,\dots,X_n$.
\begin{enumerate}
\item[(a)] The set $\calZ(S) = \{ (a_1,\dots,a_n)\in\FF_2^n \mid f(a_1,\dots,a_n)=0$ for all $f\in S\}$
is called the \textbf{zero set} of~$S$.

\item[(b)] The set $\calS(F) = \{(a_1,\dots,a_n)\in\FF_2^n \mid F|_{X_1\mapsto a_1,\dots,X_n\mapsto a_n} \equiv \True\}$
is called the \textbf{set of satisfying assignments} of~$F$.

\item[(c)] Given an ideal~$I$ in~$\BB_n$, a propositional logic formula~$F$ is called 
a \textbf{logical representation} of~$I$ if $\calS(F)=\calZ(I)$.

\item[(d)] Given a propositional logic formula~$F$, the uniquely determined ideal~$I_F$ 
in~$\BB_n$ such that $\calZ(I_F) = \calS(F)$ is called the \textbf{algebraic representation} of~$F$.
\end{enumerate}
\end{definition}

Recall that $\BB_n$ is a principal ideal ring in which every ideal has a unique generator and that every
propositional logic formula is equivalent to a formula in \textbf{conjunctive normal form (CNF)}.
%In this way the definition provides a bijection between Boolean polynomials and propositional logic
%formulas in~CNF.
Effective transformations between these representations have been studied extensively
(see for instance~\cite{CSCM} and~\cite{JK}).

One disadvantage of converting Boolean polynomials to CNF is that sums correspond to XOR connectives
and a long chain of XOR connectives yields an exponentially large set of CNF clauses.
To address this problem, we introduce a new type of normal form next. Afterwards, we examine
algorithms for converting Boolean polynomials to this normal form and back.

\begin{definition}[XOR-OR-AND Normal Form]\label{def:xnf}
Let $F$ be a propositional logic formula.
\begin{enumerate}
\item[(a)] A formula of the form $L_1 \oplus L_2 \oplus \cdots \oplus L_m$ with literals $L_1,\dots,L_m$ is called a \textbf{lineral}.

\item[(b)] A disjunction of linerals is called an \textbf{XNF clause}.

\item[(c)] The formula~$F$ is said to be in \textbf{XOR-OR-AND normal form (XNF)}
if~$F$ is a conjunction of XNF clauses.

\item[(d)] Let $k\in\NN$. If~$F$ is in XNF and every XNF clause of~$F$ involves
at most $k$ linerals then we say that~$F$ is in \textbf{$k$-XNF}.

\end{enumerate}
\end{definition}

Notice that every literal is also a lineral. Hence every formula in CNF is already
in XNF. The negation of a lineral is a lineral because of
$$
\lnot ( L_1 \oplus L_2 \oplus \cdots \oplus L_m) \;\equiv\; \lnot L_1 \oplus
L_2 \oplus \cdots \oplus L_m .
$$
Moreover, every lineral is equivalent to a lineral of the
form~$\LXOR_i X_i$ or~$\lnot (\LXOR_i X_i)$.

Observe that~\cite{NLFHB} introduces a normal form with the same name, but for 
formulas that consist of CNF clauses and XOR constraints on the variables.
In the terminology defined here, these are XNF unit clauses.
We also refer to a formula of this type as a CNF-XOR, 
consistent with related research (see~\cite{DMV,JK,LJN,SM,SNC,TID}).

The motivation for introducing the XNF is its algebraic representation
which can be described as follows.

\begin{remark}\label{rem:xnfAlg}
Let $X_1,\dots,X_n$ be propositional logic variables.
\begin{enumerate}
\item[(a)] Let $L = X_{i_1} \oplus \dots \oplus X_{i_t}$ be a lineral with $i_1,\dots,i_t\in\{1,\dots,n\}$.
Then the algebraic representation of~$L$ is the ideal $\langle x_{i_1}+\dots+x_{i_t}+1 \rangle $ in~$\BB_n$.
Thus linerals correspond to linear polynomials in~$\BB_n$.
      
\item[(b)] Let $L_1,\dots, L_s$ be linerals, and let $C = L_1 \vee \cdots \vee L_s$
be an XNF clause. For $i\in\{1,\dots,s\}$, let $\ell_i \in \LL_n$ be the algebraic representation of~$L_i$.
Then $\langle \ell_1\cdots \ell_s \rangle$ is the algebraic representation of~$C$.
Thus XNF clauses correspond to products of linear Boolean polynomials. 

\item[(c)] Let $C_1,\dots,C_r$ be XNF clauses, and let $F = C_1 \wedge \cdots \wedge C_r$
be a logical formula in XNF. For $i\in\{1,\dots,r\}$, let $c_i\in \BB_n$ be the product of linear
Boolean polynomials representing~$C_i$. Then the algebraic representation of~$F$ is
the ideal $\langle c_1,\dots,c_r\rangle$ in~$\BB_n$.
\end{enumerate}
\end{remark}

For the converse transformation, we could use the logical representation of a Boolean polynomial
which is in CNF, and hence in XNF. However, as we are striving for logical formulas in XNF
which have few and short clauses, i.e., correspond to few low-degree Boolean polynomials,
we proceed along a different path in the following two subsections.

Moreover, the XNF is the natural input to the proof system~\SRES (see~\cite{Hor}), 
and therefore builds the basis for any~\SRES-based solving algorithms.

\subsection{Reduction of Formulas in XNF to 2-XNF}

It is a well-known property of propositional logic formulas that they
can be transformed to equisatisfiable formulas in 3-CNF by introducing
additional variables. In the following we focus on an analogous transformation
of formulas in XNF.

\begin{definition}\label{def:equivalence}
Let $S \subseteq \FF_2^n$, and let $T \subseteq \FF_2^{n+m}$ for some $n,m \in \NN$.
Then we write $S \equiv_n T$ if the projection map $\pi:\; \FF_2^{n+m} \longrightarrow
\FF_2^n$ defined by $\pi((a_1,\dots,a_{n+m})) = (a_1,\dots,a_n)$ induces a
bijection $\pi\vert_T:\; T \longrightarrow S$.
\end{definition}

The relation $\equiv_n$ has the following useful properties.

\begin{remark}\label{rmk:equivTrans}
Let $F$ be a logical formula involving the variables $X_1,\dots,X_n$, and let~$G$
be a formula involving the variables $X_1,\dots,X_n,Y_1,\dots,Y_m$.
\begin{enumerate}
\item[(a)] If we have $\calS(F)\equiv_n\calS(G)$ then the formulas~$F$ and~$G$ 
are clearly equisatisfiable. More precisely, the satisfying assignments of~$G$ are 
in one-to-one correspondence with the satisfying assignments of~$F$ via the projection~$\pi$ 
to the first $n$~coordinates.

\item[(b)] In general, the relation $\equiv_n$ is not symmetric, but it is transitive 
in the following sense. Let $k,m,n \in \NN$, let $S \subseteq \FF_2^n$, let $T \subseteq 
\FF_2^{n+m}$, and let $U \subseteq \FF_2^{n+m+k}$. 
If we have $S \equiv_n T$ and $T \equiv_{n+m} U$ then $S \equiv_n U$.
\end{enumerate}
\end{remark}

The following lemma provides the key step for the reduction of formulas in XNF to 2-XNF.
It can be easily verified using a truth table.

\begin{lemma}\label{lem:2xnf1}
Let $L_1,L_2$ be two linerals, and let~$Y$ be an additional logical variable. Then we have
$$
Y \iff (L_1 \lor L_2) \;\equiv\; (Y \lor \lnot L_2) \land (\lnot(Y \oplus L_1) \lor L_2).
$$
\end{lemma}

Notice that the left side of the equivalence in this lemma is
symmetric in~$L_1$ and~$L_2$. Thus, swapping~$L_1$ and~$L_2$ on the right-hand
side of the equivalence also yields an equivalent formula.
The following Algorithm~\ref{alg:XnfTo2Xnf} converts logical formulas in XNF to 2-XNF.
\smallskip

\begin{algorithm}[ht]
  \DontPrintSemicolon
  \SetAlgoLongEnd
  \SetKwInOut{Input}{Input}
  \SetKwInOut{Output}{Output}

  \Input{A logical formula~$F$ in XNF involving $n$ variables.}
  \Output{A logical formula~$G$ in 2-XNF with $\calS(F) \equiv_n \calS(G)$.}
  \BlankLine
  Let $i = 0$ and $M = \emptyset$. Write~$F = \LAND_{k=1}^r C_k$. \;
  \For{$k=1$ \KwTo $r$}{
    \While{$C_k$ contains more than two linerals}{
        Write~$C_k = \LOR_{j=1}^s L_j$ with linerals~$L_j$. \;
        Increase~$i$ by~$1$ and let~$Y_i$ be a new variable. \;
        Replace~$C_k$ by $(Y_i \lor \LOR_{j=3}^s L_j)$. \;
        % Let~$C_k = (Y_i \lor \LOR_{j=3}^s L_j)$. \;
        Adjoin~$\{(Y_i \lor \lnot L_2), (\lnot(Y_i \oplus L_1) \lor L_2)\}$ to~$M$. \;
    }
    Append~$C_k$ to~$M$. \;
  }
  \Return $\LAND M$. \;
  \caption{\XnfToTXnf \,--\, XNF to 2-XNF Conversion}
  \label{alg:XnfTo2Xnf}
\end{algorithm}

\begin{proposition}\label{prop:2xnf}
Let $F$ be a propositional logic formula in XNF involving~$n$ logical variables. 
Then \XnfToTXnf is an algorithm which returns a logical formula~$G$ in 2-XNF such that 
$\calS(F) \equiv_n \calS(G)$.
\end{proposition}

\begin{proof} As the number of linerals in~$C_k$ is decreased with every execution of
Line~$6$, the loop in Lines~$3$-$7$ stops after finitely many iterations. Thus the 
algorithm terminates after finitely many steps.

To prove correctness, we first observe that every XNF clause which is added to~$M$
contains at most two linerals, so the output formula is indeed in 2-XNF.
Moreover, by Lemma~\ref{lem:2xnf1}, we have
$$
\calS(C_k) \;\equiv_{n+i-1}\; \calS \big(\, (Y_i \lor \textstyle{\LOR_{j=3}^s} L_j) 
    \,\land\, (Y_i \lor \lnot L_2) 
    \,\land\, (\lnot(Y_i \oplus L_1) \lor L_2)\,\big)
$$
in Line~$5$ of the algorithm.  Hence we obtain $\calS(\LAND M \land C_k) \equiv_{n+i-1} 
\calS(\LAND M' \land C_k')$ in Line~$7$, where~$M'$ and~$C_k'$ denote the values of~$M$ and~$C_k$,
respectively, after their modification inside the loop (Lines~$3$-$7$). By
Remark~\ref{rmk:equivTrans}.b, this implies~$\calS(F) \equiv_n \calS(\LAND M)$
after every iteration of the outer loop (Lines~$2$-$8$), and consequently after
its termination.
\end{proof}

\begin{example}
Consider the formula $F = X_1 \lor X_2 \lor X_3$ in 3-CNF. 
Applying \XnfToTXnf to~$F$ yields the logical formula
$$
G \;=\; (Y_1 \lor X_3) \,\land\, (Y_1 \lor \lnot X_2) \,\land\, (\lnot(Y_1 \oplus X_1) \lor X_2).
$$
where $Y_1$ is a new variable, and we have $\calS(F) \equiv_3 \calS(G)$.
\end{example}

\begin{remark}
Suppose a formula $F$ is in $k$-XNF for some $k \in \NN$ and contains $r$ XNF
clauses. Then \XnfToTXnf introduces at most $r(k-2)$ new variables and~$2r(k-2)$ new clauses, since at
most $k-2$ new variables are added for each clause. This shows that every
formula in CNF can be converted to a formula in 2-XNF in polynomial time.
Consequently, the decision problem for 2-XNF instances is NP-complete.
\end{remark}

In spite of this seemingly negative worst-case complexity, it is well-known that
2-CNF formulas can be solved in linear time (e.g., see~\cite{APT}).
In Section~\ref{sec:solving}, we will further address how some of the core ideas
of 2-CNF solving can be translated to solving formulas in 2-XNF.
Finally, note that one can not only reduce the size of the XNF clauses, 
but also the \textit{length} of its linerals, i.e., the number of variables it contains, 
by using additional variables.

\begin{remark}
Let $L,L_1,L_2$ be linerals with $L\equiv L_1\oplus L_2$.
If~$Y$ is an additional logical variable, then we have
$L \equiv (L_1\oplus Y) \land (L_2\oplus \neg Y)$.
Repeated application and addition of new variables shows that every XNF formula 
can be reduced to a 2-XNF formula in which each lineral is a XOR of at most~$3$ variables.
\end{remark}

Better constructions to trade the length of linerals with additional variables 
can be derived from the methods of~\cite{EKMS} and~\cite{NLFHB}.

\subsection{2-XNF Representations of Boolean Polynomials}

In order to apply 2-XNF solving algorithms to practical instances,
we first need to create tools to convert problems given via Boolean
polynomials in ANF to logical formulas in 2-XNF.

A straightforward approach is to search for XORs of variables in a CNF
representation of the problem which correspond to XNF clauses of size~$1$
as for instance done in~\cite{NLFHB}.
While this produces XNF instances, in many situations
it does not capture the XOR-rich information well.
In fact we should find non-trivial XNF clauses when they exist to harness the
full potential of XNF-SAT solvers.
This is why we suggest to start with an ANF representation of the problem under
investigation, as it is more compact and uses fewer variables. So, in this
section we show how Boolean polynomials can be converted to 2-XNF. To illustrate the
algorithm, we apply it to problems with a cryptographic background.

To ease the notation we switch completely to the algebraic point of view.
Not only the input of our conversion algorithm is denoted algebraically, but
also the output 2-XNF. In view of Remark~\ref{rem:xnfAlg}, the following definition
captures this approach.

\begin{definition}[2-XNF Representation]\label{def:2xnfRepr}
Let~$I$ be an ideal in~$\BB_n$.
A set of Boolean polynomials of the form 
$S = \{ f_1g_1, \dots, f_k g_k, \ell_1, \dots, \ell_s\} \subseteq \BB_{n,m}$
with $f_i,g_i,\ell_j \in \mathbb{L}_n$ is called a \textbf{2-XNF representation} of~$I$
if $\calZ(I) \equiv_n \calZ(S)$.

Similarly, a set $S \subseteq \BB_{n,m}$ as above is called a 2-XNF
representation of a Boolean polynomial $f\in\BB_n$ if~$S$ is a 2-XNF representation
of~$\langle f\rangle$.
\end{definition}
  
%Recall that ideals in~$\BB_n$ are uniquely determined by their zero sets. Hence, if
%$S$ is already contained in~$\BB_n$, then $\calZ(I) \equiv_n \calZ(S)$ if and only 
%if $I = \langle S\rangle$.
Now Proposition~\ref{prop:2xnf} immediately implies the following result.

\begin{corollary}\label{cor:poly2xnf}
Let $I$ be an ideal in~$\BB_n$. Then there exists a 2-XNF representation of~$I$.
\end{corollary}

The next proposition shows a direct way to compute 2-XNF representations of
certain polynomials. It is an algebraic formulation of Lemma~\ref{lem:2xnf1}.

\begin{proposition}\label{prop:polySubLogic}
Let $g = \ell_1\ell_2+\ell_3 \in \BB_{n}$, where $\ell_1,\ell_2,\ell_3 \in \LL_n$,
and let
$$  
  S \;=\; \{ \ell_3(\ell_2+1),\; \ell_2(\ell_1+\ell_3)\}.
$$
Then we have $\langle S\rangle = \langle g\rangle$. In particular, the set 
$S$ is a 2-XNF representation of~$g$.
\end{proposition}

\begin{proof}
From $g = \ell_3(\ell_2+1)+\ell_2(\ell_1+\ell_3)$, we obtain $g \in \langle S \rangle$ and hence $\langle g \rangle \subseteq \langle S \rangle$.
Moreover, we have $\ell_3(\ell_2+1) = (\ell_2+1)g \in \langle g\rangle $ and
$\ell_2(\ell_1+\ell_3) = \ell_2 g \in \langle g\rangle $, which implies $S \subseteq \langle g \rangle$ and hence $\langle S \rangle \subseteq \langle g \rangle$.
\end{proof}

\begin{remark}
    To see the connection with Lemma~\ref{lem:2xnf1}, let $L_1$, and $L_2$ be linerals and $Y$ be an additional variable.
    Let $\ell_1,\ell_2,\ell_3 \in \LL_n$ such that $\langle \ell_i \rangle$ is the algebraic representation of $L_i$ for $i \in \{1,2\}$ and $\langle \ell_3 \rangle$ is the algebraic representation of $Y$. Then $\langle \ell_1\ell_2+\ell_3 \rangle$ is the algebraic representation of $Y \leftrightarrow (L_1 \lor L_2)$, $\langle \ell_3(\ell_2+1) \rangle$ is the algebraic representation of $(Y \lor \lnot L_2)$, and $\ell_2(\ell_1+\ell_3)$ is the algebraic representation of $(\lnot(Y \oplus L_1) \lor L_2)$.
\end{remark}

%This proposition
Proposition \ref{prop:polySubLogic} immediately yields the following Algorithm~\ref{alg:polyTo2Xnf} for computing a
2-XNF representation of a given Boolean polynomial.
\smallskip

\begin{algorithm}[ht]
  \DontPrintSemicolon
  \SetAlgoLongEnd
  \SetKwInOut{Input}{Input}
  \SetKwInOut{Output}{Output}

  \Input{A Boolean polynomial $f \in \BB_n$.}
  \Output{A 2-XNF representation of~$f$.}
  \BlankLine
  Set $i = 0$ and $M = \emptyset$. \;
  \For{$t \in \Supp(f)$}{
    \While{$\deg(t) > 1$}{
        Increase~$i$ by~$1$ and let~$y_i$ be a new indeterminate. \;
        Write~$f = t+f'$ and $t = \ell_1\ell_2s$
        where $s$ is a term, $\ell_1,\ell_2$ are distinct indeterminates, and $f' \in \BB_{n,i-1}$. \;
        Replace $t$ by $y_is$ and $f$ by $y_is+f'$. \;
        Adjoin $\{ y_i(\ell_2+1),\; \ell_2(\ell_1+y_i)\}$ to~$M$.
    }
  }
  \Return $M \cup \{f\}$.
  \caption{\AnfToTXnf \,--\, Boolean Polynomials to 2-XNF}
  \label{alg:polyTo2Xnf}
\end{algorithm}

\begin{proposition}\label{prop:alg:polyTo2Xnf}
Let~$f \in \BB_n$. Then~$\AnfToTXnf$ is an algorithm which returns a 2-XNF
representation $S = \AnfToTXnf(f)$ of~$f$.
\end{proposition}

\begin{proof}
First we see that in each iteration of the inner loop (Lines~$3$-$7$), the
degree of~$t$ decreases by one, so it eventually reaches~1. Moreover,
the polynomial $f$ is updated in Line~$6$ in such a way that the term~$t$ 
in the support of $f$ is replaced by a term of smaller degree. Thus the outer loop (Lines~$2$-$7$)
terminates eventually and the procedure stops in Line~$8$.
In particular, at this point $f$ is linear and all elements of $M \cup \{f\}$ are
linear or products of two linear polynomials. Hence the output is in 2-XNF.
    
For the correctness, assume that we are in the $i$-th iteration of the inner
loop (Lines~$3$-$7$). Denote the values of $f$ and $M$ after the $i$-th
iteration by $f_i$ and $M_i$, respectively. Here we let $M_0=\emptyset$ and
$f_0$ denote the initial input value of~$f$. Consider the ideals
$J = \langle M_{i-1} \cup \{f_{i-1}\} \rangle \subseteq \BB_{n,i-1}$ and 
$J' = \langle M_{i-1} \cup \{f_{i-1}\} \cup \{y_i+\ell_1\ell_2\} \rangle \subseteq \BB_{n,i}$.
For $c = (c_1,\dots,c_{n+i}) \in \FF_2^{n+i}$, we see that $c \in \calZ(J')$
if and only if $(c_1,\dots,c_{n+i-1}) \in \calZ(J)$ 
and $c_{n+i} = (\ell_i\ell_j)(c_1,\dots,c_{n+i-1})$. Hence $\calZ(J) \equiv_{n+i-1} \calZ(J')$.
    
Now observe that $f_i = f_{i-1}+s(\ell_1\ell_2+y_i) \equiv f_{i-1} \mod J'$,
and hence $J' = \langle M_{i-1} \cup \{f_i\} \cup \{y_i+\ell_1\ell_2\} \rangle $.
Thus Proposition~\ref{prop:polySubLogic} shows~$J' = \langle M_i \cup \{f_i\}\rangle $.
Therefore we have~$\calZ(f_0) \equiv_n \calZ(M\cup\{f\})$ after every
iteration of the inner loop, i.e., the output in Line~$8$ is indeed a 2-XNF
representation of the input~$f$.
\end{proof}

\begin{example}
Consider the polynomial $f = x_1x_2x_3 \in \BB_3$. The ideal $\langle f\rangle$ 
is the algebraic representation of the clause $\neg X_1\lor \neg X_2\lor \neg X_3$.
We introduce a new indeterminate $y_1$ and construct the ideal
$$
I \;=\; \langle \, f, y_1+x_1x_2 \, \rangle \;=\; \langle \, y_1x_3,\, y_1(x_2+1),\, 
x_2(x_1+y_1) \,\rangle \; \subseteq \BB_{3,1}.
$$
Then we have $\calZ(f) \equiv_3 \calZ(I)$, which shows that the set 
$S = \{y_1x_3,\, y_1(x_2+1),\, x_2(x_1+y_1)\} \subseteq \BB_{3,2}$ is a 2-XNF representation of $f$.
Notice that $S$ corresponds to the 2-XNF formula
$$
(\neg Y_1 \lor \neg X_3) \,\land\, (\neg Y_1 \lor X_2)\,\land\,(\neg X_2 \lor \neg(X_1\oplus Y_1)\, )
$$
in the variables $X_1, X_2, X_3, Y_1$.
\end{example}

Notice that \AnfToTXnf employs Proposition~\ref{prop:polySubLogic} only for
replacing products of two indeterminates. For quadratic polynomials,
this uses one additional variable for every non-linear term. With the following 
optimised Algorithm~\ref{alg:quadPolyTo2Xnf}, one may replace more than one term at a time.
\smallskip

\begin{algorithm}[ht]
  \DontPrintSemicolon
  \SetAlgoLongEnd
  \SetKwInOut{Input}{Input}
  \SetKwInOut{Output}{Output}

  \Input{A Boolean polynomial $f \in \BB_n$ with $\deg(f)\le 2$.}
  \Output{A 2-XNF representation of~$f$.}
  \BlankLine
  Let $i = 0$ and $M = \emptyset$. \;
  \While{$\deg(f) = 2$}{
    Increase~$i$ by~$1$ and let~$y_i$ be a new indeterminate. \;
    Write $f = \ell_1\ell_2+f'$ for distinct $\ell_1,\ell_2 \in \LL_n$ and for $f' \in \BB_{n,i-1}$ 
    such that $\Supp(f')$ contains fewer non-linear terms than $\Supp(f)$. \;
    Set $f=y_i+f'$. \;
    Adjoin~$\{y_i(\ell_2+1),~\ell_2(\ell_1+y_i)\}$ to~$M$.
  }
  \Return $M \cup \{f\}$.
  \caption{\QAnfToTXnf \,--\, Quadratic Boolean Polynomials to 2-XNF
  }
  \label{alg:quadPolyTo2Xnf}
\end{algorithm}

\begin{proposition}\label{prop:alg:quadPolyTo2Xnf}
Let $f \in \BB_n$ with $\deg(f) \leq 2$. Then~$\QAnfToTXnf$ is an
algorithm which returns a 2-XNF representation $S = \QAnfToTXnf(f)$ of~$f$.
\end{proposition}

\begin{proof}
After each iteration of the loop (Lines~$2$-$6$), the support of~$f$
contains fewer non-linear terms. Therefore~$f$ eventually becomes linear
and the loop terminates.

For proving correctness, consider the iterations of the loop.
As in Proposition~\ref{prop:alg:polyTo2Xnf}, we see
that $\langle M_{i-1} \cup \{f_{i-1}\} \cup \{y_i+\ell_1\ell_2\}\rangle = 
\langle M_i \cup \{f_i\}\rangle $, where $M_i$ and $f_i$ denote the values of~$f$ 
and $M$ after the $i$-th iteration, and~$f_0$ is the initial value of~$f$.
In particular, this shows that we have $\calZ(f_0)\equiv_n\calZ(M\cup\{f\})$ after 
every iteration. Thus the output is a 2-XNF representation of the input~$f$.
\end{proof}

To implement Line~$4$ of \QAnfToTXnf efficiently, we may use different approaches.
The following remark collects some of them.

\begin{remark}\label{rmk:findGoodSubstitution}
Let~$f\in\BB_n$ be of degree $\le 2$. In order to find $\ell_1,\ell_2 \in \LL_n$ 
such that~$\Supp(f-\ell_1\ell_2)$ contains fewer quadratic terms than~$\Supp(f)$, we may
use one of the following methods.
\begin{enumerate}
\item[(a)]
Write $f = x_i\ell_i+g_i$ with $i \in \{1,\dots,n\}$ and $\ell_i \in \LL_n \setminus \FF_2$
such that no term in the support of~$g_i$ is divisible by~$x_i$.
Then the support of $f-x_i\ell_i = g_i$ is a proper subset of $\Supp(f)$.
In particular, it contains fewer quadratic terms.
Repeating this step requires at most~$n-1$ substitutions until all non-linear
terms in~$f$ have been replaced.
Hence any quadratic polynomial $f\in\BB_n$ has a 2-XNF representation that
uses fewer than $n-1$ additional indeterminates, even though the support of~$f$ 
may contain up to~$\binom{n}{2}$ quadratic terms.

\item[(b)]
Let $y_1,\dots,y_n, z_1,\dots,z_n$ be new indeterminates, and let
$$
G \;=\; (y_1 x_1 + \cdots + y_n x_n)\cdot (z_1 x_1 + \cdots +z_n x_n) \in \BB_n[y_1,\dots,y_n, z_1,\dots,z_n]
$$
be a product of two \textit{generic} linear Boolean polynomials. By multiplying out, we obtain a
representation 
$$
G \;=\; \tsum_{1 \le i < j \le n} G_{ij} x_i x_j + \tsum_{k=1}^n H_k x_k .
$$  
with $G_{ij},H_k\in \FF_2[y_1,\dots,y_n, z_1,\dots,z_n]$.  

Write $f = \sum_{1 \le i < j\le n} f_{ij} x_i x_j + \sum_{k=1}^n f_k x_k + f_0$ 
with $f_{ij},f_k, f_0 \in \FF_2$.     
If we find a tuple $ c = (a_1,\dots,a_n, b_1, \dots, b_n) \in \FF_2^{2n}$
such that as many of the equations $f_{ij} = G_{ij}(c)$ as possible are satisfied, 
then the linear Boolean polynomials $\ell_1 = a_1 x_1 + \cdots + a_n x_n$ and
$\ell_2 = b_1 x_1 + \cdots + b_n x_n$ satisfy the property that $f - \ell_1\ell_2$
contains as few quadratic terms in its support as possible.
Such a tuple $c$ can be found using an OMT solver, e.g., using \OptiMathSAT
(see \cite{ST}), or by rephrasing the optimization
problem as a MaxSAT problem and using an adequate solver, e.g., using \MaxHS
(see~\cite{Dav}).
\end{enumerate}

\end{remark}

The strategy of part~(b) works well on small inputs, say polynomials having fewer than 20
indeterminates. For cases involving larger numbers of indeterminates, it is better to
combine part~(a) with the next observation.

\begin{lemma}\label{lem:combiningSubs}
Let $f\in\BB_n$ and $\ell_1,\ell_2,\ell_1',\ell_2' \in \LL_n$ with $\Supp(\ell_1\ell_2) \subseteq \Supp(f)$
and $\Supp(\ell_1'\ell_2') \subseteq \Supp(f)$.
Then we have $\Supp(m_1 m_2) \subseteq \Supp(f)$ for 
$$
m_1 = \textstyle\sum(\Supp(\ell_1) \cup \Supp(\ell_1')) \quad\text{and}\quad 
m_2 = \textstyle\sum(\Supp(\ell_2) \cap \Supp(\ell_2')).
$$
\end{lemma}

\begin{proof}
Let $t = x_{i_1}x_{i_2} \in \Supp(m_1m_2)$ where $x_{i_1} \in \Supp(m_1)$ and $x_{i_2} \in \Supp(m_2)$. 
Then $x_{i_1} \in \Supp(\ell_1)$ or $x_{i_1} \in \Supp(\ell_1')$, and $x_{i_2} \in \Supp(\ell_2) \cap \Supp(\ell_2')$.
This shows $x_{i_1}x_{i_2} \in \Supp(\ell_1\ell_2)$ or $x_{i_1}x_{i_2} \in \Supp(\ell_1'\ell_2')$. 
Both imply $t \in \Supp(f)$.
\end{proof}

Using the method of Remark~\ref{rmk:findGoodSubstitution}.a, we can now find many distinct pairs 
$(\ell_1,\ell_2)\in\LL_n^2$ with $\Supp(\ell_1\ell_2)\subseteq\Supp(f)$.
Applying the Lemma randomly to two such pairs of linear polynomials, we find more pairs 
$(m_1,m_2)\in\LL_n^2$ with $\Supp(m_1m_2)\subseteq\Supp(f)$.
Repeating this procedure for some time can generate many non-trivial such pairs, 
and we can simply choose the one which eliminates the most non-linear terms.
This has proven to produce very good results, even for polynomials with a high number of indeterminates.

\begin{example}
Let us apply Algorithm~\QAnfToTXnf to the Boolean polynomial $f = x_1x_3+x_2x_3+x_1x_4+x_2x_4+x_1 \in \BB_4$.
In Line~4 we try to cancel out as many non-linear terms as possible, following the above approach. 
Using Remark~\ref{rmk:findGoodSubstitution}.a, we see that 
$\Supp(\, x_1\cdot (x_3+x_4+1) \,) \,\subseteq\, \Supp(f)$ and $\Supp(\, x_2\cdot(x_3+x_4) \,) \,\subseteq\, \Supp(f)$. 
By applying Lemma~\ref{lem:combiningSubs} with $m_1 = x_1+x_2$ and $m_2 = x_3+x_4$, 
we get $\Supp(m_1m_2) \subseteq \Supp(f)$. Let $y_1$ be a new indeterminate and write $f = m_1m_2 + x_1$.
Now we replace $f$ by $y_1+x_1$ and set
$$
M \;=\; \{\, y_1(x_3+x_4+1),\, (x_3+x_4)(x_1+x_2+y_1) \,\}.
$$
Notice that the loop now ends, as $f$ is linear, and the 2-XNF representation $\{f\} \cup M$ of $f$ is returned.
This corresponds to the 2-XNF formula
$$\left(\neg Y_1 \lor (X_3\oplus X_4)\right) \,\land\, (\neg (X_3\oplus X_4) \lor \neg (X_1\oplus X_2\oplus Y_1))\,\land\,\neg (Y_1\oplus X_1))$$
in the variables $X_1, X_2, X_3, X_4, Y_1$.
\end{example}

After discussing the effective computation of 2-XNF representations of individual
polynomials, we now turn our attention to Boolean polynomial ideals given by several generators.
In this case we can avail ourselves of the following approaches.

\begin{remark}\label{rmk:convPolySys}
Let $f_1,\dots,f_s \in \BB_n \setminus \{0\}$, and let $I = \langle f_1,\dots,f_s\rangle $.
The following methods can be applied to find a 2-XNF representation of~$I$.
\begin{enumerate}
\item[(a)]
The most basic approach is to apply \AnfToTXnf (or \QAnfToTXnf) to~$f_i$
for~$i\in\{1,\dots,s\}$ and to combine the individual 2-XNF representations to
get one for~$I$. Unfortunately, this tends to introduce more
additional variables than necessary, since the same terms in different
polynomials will be replaced with distinct additional indeterminates.

\item[(b)]
If $f_1,\dots,f_s$ are quadratic, the problem in~(a) can be counteracted as follows. 
During the computation of the 2-XNF representations of the $f_1,\dots,f_s$, 
we remember how the additional indeterminates~$y_1,\dots,y_m$ were used to substitute
products $\ell_{11}\ell_{12},\dots,\ell_{m1}\ell_{m2}$ in the execution of 
Lines~$4$-$5$ of \QAnfToTXnf.
After those individual conversions, we compute an $\FF_2$-basis
$\{h_1,\dots,h_t\} \subseteq \LL_m$ of the set of relations 
$$
\{ g \in \LL_m \mid g(\ell_{11}\ell_{12}, \dots, \ell_{m1}\ell_{m2}) = 0 \}.
$$
Then we return the union of all the individual 2-XNF representations
and~$\{h_1,\dots,h_t\}$. Each of these linear Boolean polynomials eliminates one
variable in the process of computing $\calZ(I)$.
\end{enumerate}
\end{remark}

In particular, instances coming from cryptographic attacks can be converted
efficiently using those approaches. In many ciphers the only non-linearity
appears in the so-called \textit{S-Boxes}. Usually, these involve only a small number of
indeterminates, i.e., they can be represented by relatively few non-linear polynomials in a small
number of indeterminates. To illustrate this approach, let us examine the encryption map
of the \Ascon\ cryptosystem (see~\cite{DEMS}) which has been selected for the standardization of lightweight ciphers by NIST.

\begin{example}\label{exa:ascon_sbox}
As specified in~\cite{DEMS}, the \Ascon\ cryptosystem is a 128-bit lightweight cipher.
%using 5-bit S-boxes.
% 
\begin{enumerate}
\item[(a)] Let $s\colon \FF_2^5 \to \FF_2^5$ be the 5-bit S-box used in the \Ascon\ cipher. 
Consider the Boolean polynomial ring $\BB_{5,5} = \FF_2[x_1, \dots,x_5, y_1, \dots,y_5]$ and let $I \subseteq \BB_{5,5}$ be
the vanishing ideal of the set of points $\{(a,s(a)) \mid a\in \FF_2^5\} \,\subseteq\, \FF_2^{10}$.
Using~\cite{BCP}, we know that~$I$ is generated by five quadratic polynomials in $\BB_{5,5}$. 
Applying \QAnfToTXnf together with the method from Remark~\ref{rmk:findGoodSubstitution}.b and
Remark~\ref{rmk:convPolySys}.b, we obtain a 2-XNF representation
of~$I$ consisting of~$10$ products of two linear polynomials and not a single
additional indeterminate.

\item[(b)] Altogether, these methods construct a 2-XNF representation
of the entire \Asconnb\ cipher (processing $8$ bytes of plaintext) involving as little as $6080$ variables 
and $17\,664$ clauses. 

For comparison, converting the same polynomials to CNF using \PolyBoRi
(see~\cite{BD}) requires $12\,224$ variables and $137\,739$ clauses, 
the methods from~\cite{JK} require $55\,825$ variables and $214\,024$ clauses, 
and the conversion tool within \Bosphorus (see~\cite{CSCM}) requires $49\,289$ variables
and $1\,424\,034$ clauses for the logical representation of the cipher.
\end{enumerate}
\end{example}

This shows that encoding XOR-rich formulas in 2-XNF yields far more compact 
representations than state-of-the-art conversions to sets of CNF clauses.

\begin{remark}
    To efficiently store instances in XNF, we suggest a derivation of the
    established DIMACS standard for CNFs: in the place of literals (encoded 
    as~\texttt{-L} or~\texttt{L}) we encode linerals as literals connected
    (without whitespace) with the symbol~\texttt{+}.
    Then the \Ascon\ S-Box has the following XNF-representation:
    \bigskip
    
    \noindent\hbox{\hspace{.2\textwidth}
    \begin{minipage}{.33\textwidth}
        \texttt{p xnf 10 10}\\
        \texttt{-2 4+5+6 0}\\
        \texttt{2+3 -1+2+4+5+7 0}\\
        \texttt{-1 2+3+9 0}\\
        \texttt{-2+3 1+5+7 0}\\
        \texttt{-2 1+4+10 0}
    \end{minipage}\hspace{0.04\textwidth}
    \begin{minipage}{.33\textwidth}
        \texttt{\phantom{1}}\\
        \texttt{-4 2+3+8 0}\\
        \texttt{1 -2+3+4+5+9 0}\\
        \texttt{2 -1+3+4+6 0}\\
        \texttt{2 -4+5+10 0}\\
        \texttt{4 2+3+5+8 0}
    \end{minipage}}
    \bigskip
    
    \noindent
    Note that solvers supporting this encoding can also process usual DIMACS CNF files correctly.
\end{remark}

%%%%%%%%%%%%%%%%%%%%%%%%%%%%%%%%%%%%%%%%%%%%%%%%%%%%%%%%%%%%%%%%
%
%  Section 4: 2-XNF Solving
%
%%%%%%%%%%%%%%%%%%%%%%%%%%%%%%%%%%%%%%%%%%%%%%%%%%%%%%%%%%%%%%%%

\section{Graph-based 2-XNF SAT Solving}
\label{sec:solving}

It is well-known that a satisfiable assignment of a 2-CNF instance $F$, i.e.,
a propositional logic formula in CNF where every clause has at most two literals,
can be found with linear time and space complexity (see~\cite{APT}). 
The key idea is to express the formula~$F$ by a (directed) \textit{implication graph} 
whose set of vertices is the set of literals occurring in~$F$ and their respective negations, 
and for which every clause $L_i\lor L_j$ of~$F$ corresponds to the pair of edges $(\neg L_i, L_j)$
and $(\neg L_j, L_i)$.
Then a greedy algorithm working along a \textit{topological ordering} of the
\textit{strongly connected components} of this graph constructs a satisfying
assignment.
In this section we present a graph-based 2-XNF solver that follows a DPLL
approach where the above ideas form the basis of the in-processing step.

%%%%%%%%%%%%%%%%%%%%%%%%%%%%%%%%%%%%%%
% Subsection 4.1
%%%%%%%%%%%%%%%%%%%%%%%%%%%%%%%%%%%%%%

\subsection{Implication Graph Structures}

Recall that, for a propositional logic formula~$F$ in 2-XNF, the algebraic 
representation $I_F \subseteq \BB_n$ is of the form
$$
I_F = \langle f_1 g_1, \dots, f_k g_k, \ell_1,\dots,\ell_s \rangle \subseteq \BB_n
$$
for some $f_i,g_i,\ell_j\in\LL_n$.
Based on the central idea of implication graph based linear time 2-CNF solving, we
introduce the following notion.

\begin{definition}[Implication Graph Structures]\label{def:igs}

Let $F$ be a formula in 2-XNF.
\begin{enumerate}
\item[(a)] A tuple $(L,V,E)$, where $L,V\subseteq \LL_n$ and $E\subseteq V^2$, 
is called an \textbf{implication graph structure} (\textbf{IGS}) for~$F$ if the following conditions are satisfied:
\begin{itemize}
\item[(1)] $I_F=\langle L\rangle + \langle fg \mid (f+1,g)\in E\rangle$.

\item[(2)] (Skew-Symmetry) For all $(f+1,g)\in E$, we have $(g+1,f)\in E$. 

\item[(3)] For all $f\in V$, we have $(f,f)\notin E$.
\end{itemize}

\item[(b)] Let $\sigma$ be a term ordering. An IGS $(L,V,E)$ for~$F$
is called $\sigma$-\textbf{reduced} if the polynomials in~$L$ have pairwise distinct 
leading terms and 
    $$\LT_\sigma(L)\cap\textstyle\bigcup_{f\in V}\Supp(f)=\emptyset.$$
\end{enumerate}
\end{definition}

For an IGS $(L,V,E)$, the pair $(V,E)$ is clearly a graph.
Such graphs are called \textbf{implication graphs} in view of the following observation.

\begin{remark}\label{rem:motivation}
Let~$(L,V,E)$ be an IGS for a formula~$F$, and let $(f,g)\in E$. 
By definition, we then have $(f+1)g \in I_F$, and therefore
$$
f\in I_F \quad\implies\quad   g \;=\; fg+(f+1)g \in I_F.
$$
In other words, if the source node of an edge in the graph $(V,E)$ is contained 
in the ideal $I_F$, then its target node is in~$I_F$, too.
The set~$L$ simply collects all known linear information of $I_F$.
\end{remark}

Given an IGS $G=(L,V,E)$ for a formula~$F$, a
sequence $f_1,\dots,f_s\in V$ with $(f_i,f_{i+1})\in E$ for $i\in
\{1,\dots,s-1\}$ is called a \textbf{path} in~$G$. In this case we also write 
$f_1 \to f_s$.

\begin{lemma}\label{lem:IGStransitivity}
Let $G$ be an IGS for a formula~$F$, and let 
$f\to g$ be a path in~$G$. Then we have~$(f+1)g\in I_F$.
\end{lemma}

\begin{proof}
Let the path $f\to g$ be given by $(f_i,f_{i+1})\in E$ for $i\in\{1,\dots,s-1\}$,
where $f=f_1$ and $g=f_s$ for some $s \in \NN_+$.
We show the claim by induction on~$s$. By Definition~\ref{def:igs}, the statement 
is true if $s=1$. Assume that the claim is correct for paths of length $s-1$.
Then we have $(f_2+1)f_s \in I_F$, and by Definition~\ref{def:igs} also $(f_1+1)f_2\in I_F$. 
This shows
$$ 
(f_1+1)f_{s}  \;=\; (f_1+1)f_2f_{s} + (f_1+1)(f_2+1)f_{s}\in I_F.
$$
\end{proof}

This lemma implies that the \textit{transitive closure} $(V,E')$ of~$(V,E)$ yields
an implication graph structure $(L,V,E')$ for~$F$.
It is easy to find an implication graph structure for a formula in 2-XNF, as the
next remark indicates.

\begin{remark}[Trivial Implication Graph Structures]\label{rem:trivialIGS}

Let~$F$ be a formula in 2-XNF with an algebraic representation of the form
$$
I_F \;=\; \langle f_1 g_1, \dots, f_k g_k, \ell_1, \dots, \ell_s \rangle \subseteq \BB_n
$$
where $f_i,g_i,\ell_j\in\LL_n$ are pairwise distinct.
\begin{enumerate}
\item[(a)] Then the implication graph structure $(L,V,E)$ given by $L=\{\ell_1,\dots,\ell_s\}$,
$$
V  \;=\; \bigcup_{i=1}^k\, \{\,f_i,\, f_i+1,\, g_i,\, g_i+1 \,\},\hbox{\quad and\quad}
E  \;=\; \bigcup_{i=1}^k\, \{\; (f_i+1,g_i),\, (g_i+1,f_i)\;\}
$$
is called the \textit{trivial implication graph structure} for~$F$.

\item[(b)] The implication graph structure $(L,V,E)$ given by $L=\{\ell_1,\dots,\ell_s\}$,
\begin{align*}
V  &\;=\; \{\,f_i,\, f_i+1,\, g_i,\, g_i+1,\, f_i+g_i,\, f_i+g_i+1\; \mid 1\leq i\leq k\},\hbox{\ and}\\
E  &\;=\; \bigcup_{i=1}^k \{\; (f_i+1,g_i),\, (f_i+1,f_i+g_i+1),\, (f_i+g_i, g_i) \;\} \\
   &\qquad\qquad  \cup\bigcup_{i=1}^k \{\; (g_i+1,f_i),\, (f_i+g_i,f_i),\, (g_i+1,f_i+g_i+1) \;\}
\end{align*}
is called the \textit{extended trivial implication graph structure} for~$F$.

\end{enumerate}

In both cases the size of the graph $(V,E)$ is linear in the input size of the formula~$F$,
because we have $\#V \le 6k$ and $\#E\le 3k$.
\end{remark}

\begin{example}\label{ex:igs}
Let~$F$ be a formula in 2-XNF with algebraic representation
\begin{align*}
    I_F = \Big\langle\,
    &(x_1+1)x_2,\,
      (x_2+1)(x_1+x_3),\, 
      (x_2+1)x_4,\, 
      (x_5+x_2+1)(x_1+x_3),\,  \\
    &(x_1+x_3+1)(x_1+x_2+x_3+1),\,
      (x_4+1)x_3,\, 
      (x_5+1)x_4
      \,\Big\rangle \;\subseteq\; \BB_5.
\end{align*}
Then the trivial IGS of $F$ is $(L_0,V_0,E_0)$ where $L_0 = \emptyset$ and $(V_0,E_0)$ is the graph given in Figure~\ref{fig:exigs}.

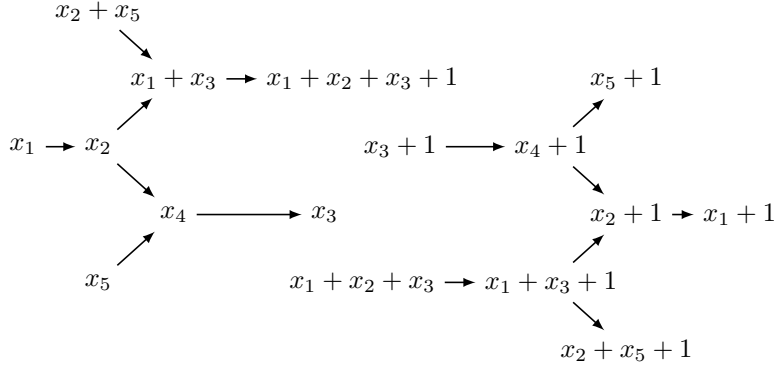
\begin{figure}
    \centering
    \tikzstyle{arrow} = [-latex,line width=0.6pt]
    \tikzstyle{revarrow} = [latex-,line width=0.6pt]
    \begin{tikzpicture}[yscale=0.9]
        % top left component
        \node (x1) at (0,0) {$x_1$};
        \node (x2) at (1,0) {$x_2$};
        \draw[arrow] (x1) -> (x2);

        \node (x1+x3) at (2,1) {$x_1+x_3$};
        \draw[arrow] (x2) -> (x1+x3);
        \node (x5+x2) at (1,2) {$x_2+x_5$};
        \draw[arrow] (x5+x2) -> (x1+x3);
        \node (x1+x2+x3+1) at (4.5,1) {$x_1+x_2+x_3+1$};
        \draw[arrow] (x1+x3) -> (x1+x2+x3+1);
        \node (x4) at (2,-1) {$x_4$};
        \draw[arrow] (x2) -> (x4);
        \node (x5) at (1,-2) {$x_5$};
        \draw[arrow] (x5) -> (x4);
        \node (x3) at (4,-1) {$x_3$};
        \draw[arrow] (x4) -> (x3);
        
        % bottom right component
        % NOTE: arrows are reversed!
        \begin{scope}[scale=-1,shift={(-9,1)}]
          \node (x1+1) at (-0.5,0) {$x_1+1$};
          \node (x2+1) at (1,0) {$x_2+1$};
          \draw[revarrow] (x1+1) -> (x2+1);
          \node (x1+x3+1) at (2,1) {$x_1+x_3+1$};
          \draw[revarrow] (x2+1) -> (x1+x3+1);
          \node (x5+x2+1) at (1,2) {$x_2+x_5+1$};
          \draw[revarrow] (x5+x2+1) -> (x1+x3+1);
          \node (x1+x2+x3) at (4.5,1) {$x_1+x_2+x_3$};
          \draw[revarrow] (x1+x3+1) -> (x1+x2+x3);
          \node (x4+1) at (2,-1) {$x_4+1$};
          \draw[revarrow] (x2+1) -> (x4+1);
          \node (x5+1) at (1,-2) {$x_5+1$};
          \draw[revarrow] (x5+1) -> (x4+1);
          \node (x3+1) at (4,-1) {$x_3+1$};
          \draw[revarrow] (x4+1) -> (x3+1);
        \end{scope}
    \end{tikzpicture}
    
    \caption{Implication graph $(V_0,E_0)$ from Example~\ref{ex:igs}.}
    \label{fig:exigs}
\end{figure}
\end{example}

Our solving algorithm starts with such a trivial IGS
for~$F$ and improves it gradually by propagation, in-processing and
guessing until we arrive at an IGS with an empty graph, i.e.,
a case where the corresponding ideal is generated just by linear polynomials.
Given that the guesses were correct, a satisfying assignment of~$F$ can then be
deduced immediately from a solution to the corresponding system of linear
equations. This improvement is measured in terms of the size of the linear part~$L$ and 
in the size of the graph $(V,E)$. The following relation specifies this in detail.

\begin{definition}\label{def:IGSordering}
Let $F$ be a formula in 2-XNF, and let $G'=(L',V',E')$ as well as $G=(L,V,E)$ be two
implication graph structures for~$F$. 
Then we write $G' \preceq G$ if and only if $\langle L'\rangleF \supseteq \langle L\rangleF$ 
and $\#V' \leq \#V$. 
Moreover, if one of the two conditions is strict, we write $G' \prec G$.
\end{definition}

The relation~$\preceq$ defines a partial quasi-ordering on the set of all
implication graph structures, i.e., it is reflexive, transitive, and by the following lemma it
satisfies the descending chain condition. The latter property is the
key ingredient for proving finiteness of the upcoming algorithms.

\begin{lemma}[Descending Chain Condition for Implication Graph Structures]\label{lem:DCC}
Let $F$ be a formula in 2-XNF. Then there is no infinite, strictly descending
chain of implication graph structures for~$F$.
\end{lemma}
\begin{proof}
For a contradiction, assume there is an infinite strictly descending chain
$(L_1,V_1,E_2) \succ (L_2,V_2,E_2) \succ \cdots$ of implication graph structures for~$F$. 
By definition, it follows that there is an ascending chain of
subspaces $\langle L_1\rangleF \subseteq \langle L_2\rangleF \subseteq \cdots$ in~$\BB_n$.
Since $\BB_n$ is a finite-dimensional $\FF_2$-vector space, this chain
becomes eventually stationary, i.e., there exists a number $k\in\NN_+$ such that
$\langle L_k\rangleF = \langle L_i \rangleF$ for all $i\geq k$. 
By Definition \ref{def:IGSordering}, this implies $\#V_{i+1} < \#V_i$ for all $i \geq k$. In this way, the $\#V_i$ form a decreasing sequence in $\NN$ which eventually becomes stationary. Consequently, at some point in the sequence, we have $\langle L_{i+1} \rangleF = \langle L_{i} \rangleF$ and $\#V_{i+1} = \#V_i$, i.e.~the chain is not strictly decreasing.
\end{proof}

To conclude this section, we present the updating Algorithm~\ref{alg:updating} of our solver
which computes the $\sigma$-reduction of any given IGS.
The method is an adaption of \textit{Gau{\ss}ian Constraint Propagation} (see~\cite[Algorithm 5.7]{Hor}) 
to implication graph structures. Note that Gau{\ss}ian Constraint Propagation itself is a 
generalization of \textit{Boolean Constraint Propagation}, also known as \textit{Unit Propagation}, 
in traditional CNF-based SAT solvers.
\smallskip

\begin{algorithm}[ht]
  \DontPrintSemicolon
  \SetAlgoLongEnd
  \SetKwInOut{Input}{Input}
  \SetKwInOut{Output}{Output}
  \Input{An IGS~$G$ for a formula~$F$, a term ordering $\sigma$.}
  \Output{A $\sigma$-reduced IGS~$G'$ for~$F$ with $G'\preceq G$.}
  \BlankLine
  Write $G=(L,V,E)$ and $\LT_\sigma$-interreduce $L$.\;
  Let $(L',V',E')=(L,\emptyset,\emptyset)$.\;
  \For{$(f,g)\in E$}{
    Let $f'=\NR_\sigma(f,L)$ and $g'=\NR_\sigma(g,L)$.\;%\tcp*{$(f'+1)\cdot g'\in I_F$}
    \lIf{$f'=0$ {\bf and} $g'\neq 0$}{append $g'$ to $L'$.}
    \lIf{$g'=1$ {\bf and} $f'\neq 1$}{append $f'+1$ to $L'$.} %actually not necessary due to skew-symmetry!
    \If{$f'\notin\FF_2$ {\bf and} $g'\notin\FF_2$ {\bf and} $f'\neq g'$}{append $(f',g')$ to $E'$, 
        append $f'$ and $g'$ to~$V'$.}
  }
  \lIf{$L\neq L'$}{set $(L,V,E)=(L',V',E')$ and go to Line~$2$.}
  \lElse{\Return $(L',V',E')$.}
  \caption{\GCP \,--\, Graph Gau{\ss}ian Constraint Propagation}
  \label{alg:updating}
\end{algorithm}

\begin{proposition}\label{prop:alg:GCP}
Let~$\sigma$ be a term ordering, let~$F$ be a formula in 2-XNF, and let~$G$
be an IGS for~$F$.
Then \GCP is an algorithm which returns a $\sigma$-reduced implication
graph structure $G'=\GCP_\sigma(G)$ for~$F$  such that $G' \preceq G$.
\end{proposition}

\begin{proof}
Since $(L,V,E)$ is an IGS for~$F$, we have $L\subseteq I_F$ and
$(f+1)g\in I_F$ for every pair $(f,g)\in E$.
Thus we see that $(f'+1)g' = (\NR_\sigma(f,L)+1) \NR_\sigma(g,L) \in I_F$ 
holds in Line~$4$.

For $f'=0$, this yields $g'\in I_F$, and for $g'=1$, we get $f'+1\in I_F$. 
For all other cases, where $f'\in\FF_2$, $g'\in\FF_2$, or $f'=g'$, we have
$(f'+1)g' = 0$, and the corresponding edge can be ignored.
This shows that after Lines~$3$-$8$ have been executed, the tuple $G'=(L',V',E')$ 
is indeed an IGS for~$F$. Moreover, $G'$ is $\sigma$-reduced, 
because for all $f'\in V'$ we have~$\LT_\sigma(L)\cap\Supp(f')=\emptyset$ by construction.

Finally, note that we always have $\langle L' \rangleF \supseteq \langle L \rangleF$ and $\#V' \leq \#V$. 
This implies $G' \preceq G$ for every iteration of Lines~$2$-$8$, and this relation 
is strict if $L \neq L'$. By Line~$9$, these steps are repeated as long as this is the case, and the
implication graph structures $(L',V',E')$ resulting from these iterations form a
strictly descending chain.
By Lemma~\ref{lem:DCC}, this chain must be finite, i.e., there can
only be finitely many iterations, and the procedure has to terminate in Line~$10$.
\end{proof}

%%%%%%%%%%%%%%%%%%%%%%%%%%%%%%
%  Subsection 4.2
%%%%%%%%%%%%%%%%%%%%%%%%%%%%%%

\subsection{Pre-Processing Techniques}

In this subsection we present two results which allow us to deduce new information from
a given implication graph structure. The first one derives new linerals,
and the second one finds new edges between the vertices of a given implication
graph. These techniques are computationally rather expensive and should be seen
as pre-processing techniques which are only applied once before the main solving procedure.

\begin{definition}
Let $F$ be a formula in 2-XNF, and let $(L,V,E)$ be an IGS for~$F$. 
\begin{enumerate}
\item[(a)] The \textbf{set of descendants} of a vertex $f\in V$ is defined by
$$
D_f \;=\;  \{f\} \cup \{g\in V \mid \; \text{there is a path } f\to g \text{ in } (V,E)\}.
$$
Note that we consider~$f$ as a descendant of itself, since we have
$(f+1)f=0\in I_F$. 

\item[(b)] The vector space $\Delta_f = \langle D_f \rangleF \subseteq \LL_n$ will be called
the \textbf{space of descendants} of~$f$.
\end{enumerate}
\end{definition}

Note that, for a vector subspace~$U$ of $\LL_n$, we let
$1+U = \{f+1 \mid f\in U\}$ be the affine subspace of~$\LL_n$ representing the
\textit{negation} of~$U$.
The space of descendants of $f\in V$ has the following useful properties. 

\begin{proposition}\label{prop:preprocessing}
Let $F$ be a formula in $2$-XNF, let $(L,V,E)$ be an implication graph
structure for~$F$.
\begin{enumerate}
\item[(a)] For all $f\in V$ and $g\in \Delta_f$, we have $(f+1)g \in I_F$.

\item[(b)] Let $f,g\in V$. If $\Delta_f \cap (1+\Delta_g) \ne \emptyset$
then $(f+1)(g+1)\in I_F$.

\item[(c)] For all $f\in V$, we have $\Delta_f \cap\Delta_{f+1} \subseteq I_F$.
\end{enumerate}
\end{proposition}

\begin{proof}
To show~(a), let $g \in \Delta_f$. We write $g = \sum_{j=1}^k g_j$ with $g_j \in D_f$. 
By Lemma~\ref{lem:IGStransitivity}, we have $(f+1)g_j \in I_F$ for $j\in\{1,\dots,k\}$. 
Hence we obtain $(f+1)g = \sum_{j=1}^k (f+1)g_j \in I_F$.

To prove~(b), let $h \in \Delta_f$ and $h+1\in \Delta_g$. Then~(a) implies
$(f+1)h\in I_F$ and $(g+1)(h+1)\in I_F$. This shows
$$ 
(f+1)(g+1) = (f+1)h(g+1)+(f+1)(h+1)(g+1) \in I_F. 
$$

For the proof of~(c), let $g \in \Delta_f \cap \Delta_{f+1}$. 
From~(a) we get $(f+1)g\in I_F$ and $fg\in I_F$, and therefore $g = fg+(f+1)g \in I_F$.
\end{proof}

\begin{example}\label{ex:igs2}
In the situation of Example \ref{ex:igs}, we have $x_1+x_2 \in \Delta_{x_2} \cap \Delta_{x_2+1}$. 
Proposition~\ref{prop:preprocessing}.c then implies $x_1+x_2\in I_F$, and thus $(L_0\cup\{x_1+x_2\},V_0,E_0)$ is an IGS for~$F$ as well.
Let $\sigma=\lex$ and apply $\GCP_\lex$ to this tuple to get an IGS $(L_1,V_1,E_1)$ for~$F$ where $L_1=\{x_1+x_2\}$ and $(V_1,E_1)$ is graph given in Figure~\ref{fig:exigs2}.
\end{example}

\begin{figure}
    \centering
    \tikzstyle{arrow} = [-latex,line width=0.6pt]
    \tikzstyle{revarrow} = [latex-,line width=0.6pt]
    \begin{tikzpicture}[yscale=0.9]
      % top left component
      \node (x2) at (1,0) {$x_2$};
      \node (x1+x3) at (2,1) {$x_2+x_3$};
      \draw[arrow] (x2) -> (x1+x3);
      \node (x5+x2) at (1,2) {$x_5+x_2$};
      \draw[arrow] (x5+x2) -> (x1+x3);
      \node (x4) at (2,-1) {$x_4$};
      \draw[arrow] (x2) -> (x4);
      \node (x5) at (1,-2) {$x_5$};
      \draw[arrow] (x5) -> (x4);
      \node (x3) at (4,-1) {$x_3$};
      \draw[arrow] (x4) -> (x3);

      % bottom right component
      % NOTE: arrows are reversed!
      \begin{scope}[scale=-1,shift={(-8,0)}]
        \node (x2+1) at (1,0) {$x_2+1$};
        \node (x1+x3+1) at (2,1) {$x_2+x_3+1$};
        \draw[revarrow] (x2+1) -> (x1+x3+1);
        \node (x5+x2+1) at (1,2) {$x_5+x_2+1$};
        \draw[revarrow] (x5+x2+1) -> (x1+x3+1);
        \draw[revarrow] (x1+x3+1) ->node (ellipse2) {} (x3);
        \node (x4+1) at (2,-1) {$x_4+1$};
        \draw[revarrow] (x2+1) -> (x4+1);
        \node (x5+1) at (1,-2) {$x_5+1$};
        \draw[revarrow] (x5+1) -> (x4+1);
        \node (x3+1) at (4,-1) {$x_3+1$};
        \draw[revarrow] (x4+1) -> (x3+1);
        \draw[arrow] (x1+x3) -> (x3+1);
      \end{scope}
    \end{tikzpicture}
    \caption{Implication graph $(V_1,E_1)$ from Example~\ref{ex:igs2}.}
    \label{fig:exigs2}
\end{figure}
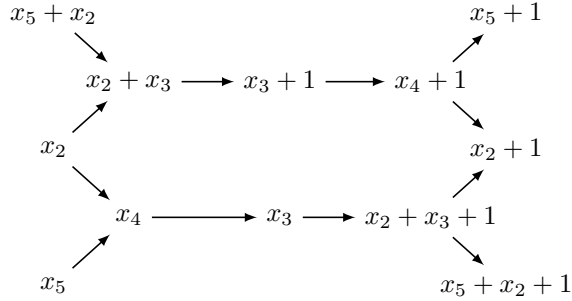

Using this proposition, we construct the following straightforward pre-processing 
Algorithm~\ref{alg:preprocessing}. It runs in polynomial time in the size of~$F$
and can find new linear information as well as new edges.
Notice that numerous intersections of affine $\FF_2$-subspaces of $\FF_2^n$ 
have to be computed.
\smallskip

\begin{algorithm}[ht]
  \DontPrintSemicolon
  \SetAlgoLongEnd
  \SetKwInOut{Input}{Input}
  \SetKwInOut{Output}{Output}
  \Input{An IGS~$G$ for a formula~$F$, a term ordering~$\sigma$.}
  \Output{A $\sigma$-reduced IGS~$G'$ for~$F$.}
  \BlankLine
  Let $(L',V',E')=\GCP_\sigma(G)$, let $L_{\mathtt{pp}}=\emptyset$ and let $E_{\mathtt{pp}}=\emptyset$.\;
  \For{$f\in V$}{
    Add a basis of $\Delta_f \cap \Delta_{f+1}$ to $L_{\mathtt{pp}}$.\;
    \For{$g\in V\setminus\{f+1\}$}{
        \lIf{$\Delta_f \cap (1+\Delta_g) \ne \emptyset$}{add $(f+1,g)$ and $(g+1,f)$ to $E_{\mathtt{pp}}$}
    }
  }
  \lIf{$L_{\mathtt{pp}}\neq\emptyset$ or $E_{\mathtt{pp}}\neq \emptyset$}{set $G=(L'\cup L_{\mathtt{pp}},V', E'\cup E_{\mathtt{pp}})$ 
      and go to Line~$1$.}
  \lElse{\Return $(L',V',E')$}
  \caption{\PP \,--\, (Edge-Extending) Pre-Processing}
  \label{alg:preprocessing}
\end{algorithm}\

It is clear that there is room for optimization of this algorithm if $G$ does not contain any
cycles. In this case it suffices to check whether $\Delta_f \cap
(1+\Delta_g) \ne \emptyset$ (see Line~$5$) initially only for sources $f,g\in V$ of~$G$, 
i.e., for vertices with no incoming edges. Only if those spaces have a non-empty
intersection, we need to consider their corresponding descendants. (This follows
immediately from the fact that $D_g\subseteq D_f$ if there is a path $f\to g$.)
Even with this optimization, finding new edges is still computationally quite
expensive. Hence Lines~$4$ and~$5$ are skipped in our implementation by default.

%%%%%%%%%%%%%%%%%%%%%%%%%%%%%%%%%%%%%
%  Subsection 4.3
%%%%%%%%%%%%%%%%%%%%%%%%%%%%%%%%%%%%%

\subsection{In-Processing Techniques}

Next we introduce two algorithms which deduce new linear polynomials from a
given implication graph structure more efficiently. Therefore they are suited as
default in-processing techniques during the main solving procedure.
In particular, the methods we look at here are (partial) generalizations of
the notions of \textit{equivalent} and \textit{failed} literals, as discussed in~\cite{HMB}.

As usual for directed graphs $G=(V,E)$, a subset $S\subseteq V$ is called a
\textbf{strongly connected component (SCC)} of~$G$ if, for all $f,g\in S$, 
there is a path $f\to g$ in~$G$ and if~$S$ is maximal with this property.
It is well-known that for any directed graph, the set of all SCCs can be computed
in $\calO(\#V+\#E)$ space and time (see~\cite{Tar}). The following
proposition indicates how these components can be used to deduce new linear
information.

\begin{proposition}\label{prop:alg:SCC}
Let~$F$ be a formula in 2-XNF, and let $G=(L,V,E)$ be an IGS
for~$F$. Denote the set of SCCs of~$(V,E)$ by~$\mathcal{C}$.
\begin{enumerate}
\item[(a)] Let $\{f_1,\dots,f_r\} \in \mathcal{C}$. Then $f_1+f_i\in I_F$
for $i\in\{1,\dots,r\}$.

\item[(b)] If $\#\mathcal{C}$ is odd, then $F$ is unsatisfiable, i.e., 
we have $I_F=\langle 1\rangle$.
\end{enumerate}
\end{proposition}

\begin{proof}
Due to the skew-symmetry of implication graph structures, for
every strongly connected component $S=\{f_1,\dots,f_r\}\in\mathcal{C}$ also
$S+1=\{f_1+1,\dots,f_r+1\}$ is an SCC of $(V,E)$. 

To show~(a), we let $i\in\{1,\dots,r\}$ and note that $f_1,f_i\in S$ implies that
there are paths $f_1\to f_i$ and $f_i\to f_1$ in~$G$. By the skew-symmetry,
we get $f_1+1\to f_i+1$. This shows $f_1,f_i\in D_{f_1}$ and $f_1+1,f_i+1\in D_{f_1+1}$, 
and hence $f_1+f_i\in \Delta_{f_1}$ as well as $f_1+f_i=(f_1+1)+(f_i+1)\in \Delta_{f_1+1}$. 
By Proposition~\ref{prop:preprocessing}.c, we thus have $f_1+f_i\in I_F$.

For the proof of~(b), notice that we can write $\mathcal{C} =\{S_1,\dots,S_c,S_1+1,\dots,S_c+1\}$ 
for some $c\in\NN$, where we have $S_i\neq S_j$ and $S_i\neq S_j+1$ for $i\ne j$. 
If $\#\mathcal{C}$ is odd, there exists an index $i\in\{1,\dots,r\}$ with $S_i = S_i+1$. 
For $f\in S_i$, we then have $f+1\in S_i+1=S_i$ and thus $1 = f+(f+1)\in I_F$ by~(a).
\end{proof}

By repeatedly computing all linear polynomials resulting from the strongly connected
components and propagating them using~\GCP, one can update a given implication
graph structure $(L,V,E)$ such that it contains no cycles, i.e., such that
$(V,E)$ becomes a directed acyclic graph (DAG). 
This is important, as for many graph-theoretic problems there are linear time
algorithms if the underlying graph is a~DAG. Algorithm~\ref{alg:eGCP} performs these 
updates.
\smallskip

\begin{algorithm}[ht]
  \DontPrintSemicolon
  \SetAlgoLongEnd
  \SetKwInOut{Input}{Input}
  \SetKwInOut{Output}{Output}
  \Input{An IGS~$G$ for a formula~$F$, a term ordering~$\sigma$.}
  \Output{An acyclic $\sigma$-reduced IGS~$G'$ for~$F$ with~$G'\preceq G$.}
  \BlankLine
  Compute $(L',V',E')=\GCP_\sigma(G)$ and let $L_\SCC = \emptyset$.\;
  Compute the set~$\mathcal{C}$ of all strongly connected components of~$(V',E')$.\;
  \lIf{$\#\mathcal{C}$ is odd}{\Return $(\LL_n,\emptyset,\emptyset)$.}
  \For{$S\in\mathcal{C}$}{
    Write $S=\{f_1,\dots,f_r\}$ and for all $i\in\{2,\dots,r\}$ append $f_1+f_i$ to $L_\SCC$.\;
  }
  \lIf{$L_\SCC\neq \emptyset$}{set $G=(L'\cup L_\SCC,V',E')$ and go to Line~$1$.}
  \lElse{\Return $(L',V',E')$.}
  \caption{\eGCP \,--\, cycle-removing GGCP}
  \label{alg:eGCP}
\end{algorithm}

\begin{proposition}\label{prop:alg:eGCP}
Let $\sigma$ be a term ordering, let~$F$ be a formula in 2-XNF, and let~$G$
be an IGS for~$F$.
Then $\eGCP$ is an algorithm which returns a tuple $G'=\eGCP_\sigma(G)$ with
the following properties.
\begin{enumerate}
\item[(a)] The tuple $G'=(L',V',E')$ is a $\sigma$-reduced implication 
graph structure for~$F$.

\item[(b)] We have $G'\preceq G$.

\item[(c)] The graph $(V',E')$ is acyclic.
\end{enumerate}
\end{proposition}

\begin{proof}
First note that if the procedure terminates in Line~$3$, the output is
correct by~Proposition~\ref{prop:alg:SCC}.b. Thus we may assume
that the procedure does not terminate in Line~3.

The tuples $(L',V',E')$ and $(L'\cup L_\SCC,V',E')$ in Lines~1 and~6 are
implication graph structures for~$F$ with $(L',V',E')\preceq (L,V,E)$ by
Propositions~\ref{prop:alg:GCP} and~\ref{prop:alg:SCC}. 
Moreover, if $L_\SCC \ne \emptyset$ then it contains at least one linear
polynomial~$f$ with $\LT_\sigma(f) \notin \LT_{\sigma}(L')$, as $(L',V',E')$
is a $\sigma$-reduced IGS. This shows $(L'\cup
L_\SCC, V', E') \prec (L',V',E')$.
    
Next we observe that the repeated iterations of Lines~$1$-$6$ yield a strictly descending
chain of IGSs which has to become stationary after
finitely many steps by Lemma~\ref{lem:DCC}. Therefore we eventually
have $L_\SCC=\emptyset$, and the procedure terminates in Line~$7$. In that
case, the graph $(V',E')$ cannot contain any cycles, as otherwise there would be a
strongly connected component, and hence Line~5 would create elements in $L_\SCC$.

Finally, note that $(L',V',E')$ is $\sigma$-reduced by
Proposition~\ref{prop:alg:GCP} and the fact that this
tuple is not changed in the last iteration of Lines~2-6.
\end{proof}

As a second in-processing technique, we adapt the concept of \textit{failed literals}, 
as discussed in~\cite{HMB}, to our more general setting.

\begin{definition}\label{def:failed}
Let $F$ be a formula in 2-XNF, and let $G=(L,V,E)$ be an implication graph structure for~$F$.
\begin{enumerate}
\item[(a)] A vertex $f\in V$ is called a \textbf{failed lineral} of~$G$ if $1\in \Delta_f$.

\item[(b)] A failed lineral $f\in V$ of~$G$ is called \textbf{trivial} 
if there is an element $g\in V$ with $f\to g$, and with $f\to g+1$ or $f\to f+1$.

\end{enumerate}
\end{definition}

These literals are of interest for in-processing, if they can be found
efficiently, since for every failed lineral we learn a new linear polynomial
in $I_F$ in the following way.

\begin{lemma}\label{lem:FXL}
Let $F$ be a formula in 2-XNF, and let $G = (L,V,E)$ be an IGS
for~$F$. If $f\in V$ is a failed lineral of~$G$, then $f+1\in I_F$. 
\end{lemma}

\begin{proof}
Let $f$ be a failed lineral of~$G$. Then $1\in \Delta_f$ yields
$f+1\in \Delta_f$. Using Proposition~\ref{prop:preprocessing}.c and
$f+1\in \Delta_{f+1}$, we get $f+1\in \Delta_f\cap\Delta_{f+1} \subseteq I_F$.
\end{proof}

To find a failed lineral, it is sufficient to check whether the vector subspace
$\Delta_f$ contains the constant polynomial~$1$. This can be done for instance by
computing the row-echelon form of a matrix in~$\FF_2^{\#D_f\times (n+1)}$. Thus we obtain
an in-processing algorithm which runs in polynomial time and space. 
However, trivial failed linerals can be found in
near-linear time, as the next remark indicates.

For an implication graph structure $(L,V,E)$ for a formula~$F$ in 2-XNF, we denote 
the \textbf{set of ancestors} of a vertex $f\in V$ by
$$ 
A_f \;=\; \{ f \} \cup \{g\in V \mid \text{there is a path } g\to f \text{ in } (V,E)\}.
$$

\begin{remark}\label{rem:FXL_commonancestors}
Let $F$ be a formula in 2-XNF, and let $(L,V,E)$ be an acyclic implication graph structure for~$F$.
\begin{enumerate}
\item[(a)] For every $g\in V$, all \textit{common} ancestors of~$g$ and~$g+1$, 
i.e., the elements of $A_g \cap A_{g+1}$, are trivial failed linerals by
definition. Conversely, every trivial failed lineral $f\in V$ is
contained in $A_g\cap A_{g+1}$ for some $g\in V$. Thus the set
$\bigcup_{g\in V} (A_g\cap A_{g+1})$ consists exactly of the trivial
failed linerals of~$G$.

\item[(b)] If $g\in V$ is a failed lineral then every $f\in A_g$ is a
failed lineral as well, since $D_g\subseteq D_f$. Thus, instead of
searching for \textit{all} common ancestors of $g$ and $g+1$, it
suffices to find the so-called \textit{lowest common ancestors},
i.e., the vertices $f\in V$ such that no out-neighbour of~$f$ is a common
ancestor of both~$g$ and~$g+1$.

\item[(c)] For sparse graphs, one of the lowest common ancestors of two
vertices can be found in constant time after a near-linear time pre-processing phase, 
see~\cite{DSHC}. This produces many, but in general not all, trivial failed linerals 
rather quickly under the assumption that the graph $(V,E)$ is sparse (see
Remark~\ref{rem:trivialIGS}).

\end{enumerate}
\end{remark}

Our implementation does not feature this advanced method for finding trivial
failed linerals, as the corresponding algorithms seem hard to implement.
Instead we resort to the following simple Algorithm~\ref{alg:FLS} which can be implemented
using only breadth-first-searches (BFS). Moreover, unlike the method of the previous
remark, it computes all trivial failed linerals.
\smallskip

\begin{algorithm}[ht]
  \DontPrintSemicolon
  \SetAlgoLongEnd
  \SetKwInOut{Input}{Input}
  \SetKwInOut{Output}{Output}
  \Input{An acyclic IGS~$G$ for a 2-XNF formula~$F$.}
  \Output{All trivial failed linerals $L_{\mathtt{TF}}$ of~$G$.}
  \BlankLine
  Write $G=(L,V,E)$, let $M=\emptyset$, and let $L_{\mathtt{TF}}=\emptyset$.\;
  Compute the set $S$ of sources of $(V,E)$.\;
  \For{$s\in S$}{
    \lIf{$s+1\in D_s$}{append $(s,s+1)$ to $M$.}
    \lFor{all $g\in V$ with $s\to g$ and $s\to g+1$}{append $(s,g)$ to $M$.}
  }
  \For{$(s,g)\in M$}{
    Append all common ancestors of $g$ and $g+1$ to $L_{\mathtt{TF}}$.
  }
  \Return $L_{\mathtt{TF}}$ 
  \caption{\FLS \,--\, Trivial Failed Lineral Search}
  \label{alg:FLS}
\end{algorithm}

\begin{proposition}\label{prop:alg:FLS}
Let $F$ be a formula in $2$-XNF, and let $G$ be an acyclic implication graph structure for~$F$. 
Then \FLS\ is an algorithm which returns a set $L=\FLS(G)$
containing all trivial failed linerals of~$G$ satisfying $1+L \subseteq I_F$.

In particular, the algorithm can be implemented to run in $\calO(\#S\cdot(\#V+\#E))$ 
time and space, where~$S$ is the set of sources in~$(V,E)$.
\end{proposition}

\begin{proof}
The finiteness of the procedure is clear, since the graph $(V,E)$ is finite. 
The correctness follows from Remark~\ref{rem:FXL_commonancestors} and the
fact that in Line~$7$ the elements of $L_{\mathtt{TF}}$ are exactly the common
ancestors of vertices~$g$ and $g+1$ for all $g\in V$. 
Now $L+1\subseteq I_F$ follows immediately from Proposition~\ref{lem:FXL}. 
The claimed run-time complexity is a consequence of the observation $\#M\leq \#S$ 
and the facts that Lines~$4$-$5$ can be implemented by a single BFS starting at~$s$, 
and that Line~$7$ amounts to two BFSs starting at~$g$ and~$g+1$ on the graph with reversed
edges.
\end{proof}

To end this section we remark that our pre-processing algorithm~\PP\ is superior
to our in-processing methods in that it learns at least the same linear
information, but might also increase the number of edges of the implication graph.

\begin{remark}
Let $\sigma$ be a term ordering, let $F$ be a formula in $2$-XNF, let $G$ be
an IGS for~$F$, and let $G'=(L',V',E')=\PP_\sigma(G)$.
Then we have $1+\FLS(G) \subseteq \langle L'\rangleF $ and $G' \preceq \eGCP(G)$.
    
This follows immediately from the fact that $\FLS$ is based on
Lemma~\ref{lem:FXL} whose proof already shows that all failed linerals
are contained in $\Delta_f \cap \Delta_{f+1} \subseteq I_F$  for some $f\in V$. 
Thus these linerals are also found by~$\PP$. 

Similarly, Algorithm $\eGCP$ is based on Proposition~\ref{prop:alg:SCC} whose proof 
shows that all linear polynomials which can be learnt here are already contained in
$\Delta_f \cap \Delta_{f+1}$ for some $f\in V$. Once again, these linerals
are found and propagated by~$\PP$.
Altogether, we see that $\PP$ essentially emulates both $\eGCP$ and~$\FLS$. 
As a consequence, Algorithm $\PP$ also ensures that its output implication graph
structure~$G'$ is acyclic.
\end{remark}

While this shows that pre-processing with~\PP is more powerful than
in-processing with~\FLS and~\eGCP, keep in mind that it is also rather expensive
due to its polynomial runtime.

\begin{example}\label{ex:igs3}
Let $(L_1,V_1,E_1)$ be the IGS from Example~\ref{ex:igs2}, then we have $\FLS(L_1,V_1,E_1) = \{x_2\}$, since there is a path $x_2\to x_2+1$ (see Figure~\ref{fig:exigs2}). This shows that $x_2$ is a failed lineral and we get $x_2+1 \in I_F$. 
An application of $\GCP_\lex$ to $(L_1 \cup \{x_2+1\},V_1,E_1)$ yields the IGS $(L_2,V_2,E_2)$ where $L_2 = \{x_1+1,\,x_2+1\}$ and $(V_2,E_2)$ is given in Figure~\ref{fig:exigs3}.

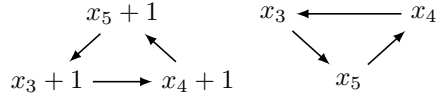
\begin{figure}
    \centering
    \tikzstyle{arrow} = [-latex,line width=0.6pt]
    \tikzstyle{revarrow} = [latex-,line width=0.6pt]
    \begin{tikzpicture}[yscale=0.9]
        \node (x3+1) at (0,0) {$x_3+1$};
        \node (x4+1) at (2,0) {$x_4+1$};
        \node (x5+1) at (1,1) {$x_5+1$};
        \draw[arrow] (x3+1) -> (x4+1);
        \draw[arrow] (x4+1) -> (x5+1);
        \draw[arrow] (x5+1) -> (x3+1);

        \begin{scope}[shift={(3,2)}]
        \node (x3) at (0,-1) {$x_3$};
        \node (x4) at (2,-1) {$x_4$};
        \node (x5) at (1,-2) {$x_5$};
        \draw[revarrow] (x3) -> (x4);
        \draw[revarrow] (x4) -> (x5);
        \draw[revarrow] (x5) -> (x3);
        \end{scope}
    \end{tikzpicture}
    \caption{Implication graph $(V_2,E_2)$ from Example~\ref{ex:igs3}.}
    \label{fig:exigs3}
\end{figure}

Notice that $(V_2,E_2)$ has two strongly connected components. Thus we can use Proposition \ref{prop:alg:SCC}.a with the SCC $\{x_3,\, x_4,\, x_5\}$ to get $x_3+x_5,\, x_4+x_5 \in I_F$.
Another application of $\GCP_\lex$ to $(L_2\cup\{x_3+x_5,\, x_4+x_5\},V_2,E_2)$ yields the IGS $(L_3,\emptyset,\emptyset)$ for~$F$ with $L_3 = \{x_1+1,\, x_2+1,\, x_3+x_5,\, x_4+x_5\}$.
By definition we now have $I_F = \langle L_3 \rangle$, i.e., a solution of~$F$ can be found by solving a system of linear equations.

Note that this is exactly the IGS that is also derived by applying~$\PP_\lex$ to $(L_0,V_0,E_0)$ from Example~\ref{ex:igs}.

\end{example}

%%%%%%%%%%%%%%%%%%%%%%%%%%%%%%%%%%%%%
%  Subsection 4.4
%%%%%%%%%%%%%%%%%%%%%%%%%%%%%%%%%%%%%

\subsection{Decision Heuristics}

Before we introduce our main DPLL-Solving Algorithm in the final subsection, we discuss
decision heuristics, i.e., methods to make \textit{good} guesses. First of all,
let us define what we precisely mean when we talk about \textit{decisions}.

\begin{definition}\label{def:decision}
Let $F$ be a formula in $2$-XNF, and let $G = (L,V,E)$ be an IGS for~$F$. 
A \textbf{decision} for~$G$ is a tuple $(L_0,L_1)$ with $L_0,L_1 \subseteq \LL_n$ such that
the following conditions are satisfied.
\begin{enumerate}
\item[(a)] $L_0\setminus \langle L\rangleF \ne \emptyset$ and
$L_1\setminus \langle L\rangleF \ne \emptyset$. 

\item[(b)] $\calZ(I_F) \subseteq \calZ(I_F+\langle L_0\rangle ) \cup 
\calZ(I_F+\langle L_1\rangle )$.
\end{enumerate}
\end{definition}

These conditions ensure that guessing either $L_0$ or $L_1$ will lead to a solution of~$F$
-- if there exists one at all. Moreover, it means that a decision $(L_0,L_1)$
for $G=(L,V,E)$ implies that $G_0=(L\cup L_0,V,E)$ and $G_1=(L\cup L_1,V,E)$
satisfy $G\succ G_0$ and $G\succ G_1$, i.e., both parts of the decisions improve
our implication graph structure.

Traditionally, CNF-based SAT solvers use decisions of the form $(\{x_i\},\, \{x_i+1\})$
or $(\{x_i+1\},\, \{x_i\})$. Our more general point of view on decisions allows us 
to guess multiple linerals at once. Before we explicitly suggest three decision heuristics,
let us consider the following general constructions.

\begin{proposition}\label{prop:decheu}
Let $\sigma$ be a term ordering, let $F$ be a formula in $2$-XNF, and let
$G=(L,V,E)$ be a $\sigma$-reduced IGS for~$F$.
\begin{enumerate}
\item[(a)] For every $f\in V$, the tuple $(D_f, D_{f+1})$ is a decision for~$G$.

\item[(b)] If $f_1\to \dots \to f_r$ is a path in~$G$  then 
$(\{f_1+f_i \mid i\in\{2\dots,r\}\},\, \{f_1+1,f_r\})$ is a decision for~$G$.
\end{enumerate}
\end{proposition}

\begin{proof}
Let $f\in V$. Then the fact that~$G$ is $\sigma$-reduced yields $f,f+1\notin \langle L\rangleF $.
For every $a \in \calZ(I_F)$, we have $f(a)=1$ or $f(a)=0$. This
shows that $(\{f\},\{f+1\})$ is a decision for~$G$. Now it suffices to note
that $I_F+\langle f\rangle  = I_F + \Delta_f = I_F + \langle D_f \rangleF$ by Proposition~\ref{prop:preprocessing}
and Remark~\ref{rem:motivation}.
    
Next we let $f_1\to\dots\to f_r$ be a path as in~(b). Then we have $f_1+f_r, f_r \notin
\langle L\rangleF $, since $G$ is $\sigma$-reduced. 
Consider a point $a\in\calZ(I_F)$. If $((f_r + 1)f_1)(a) = 1$ then $f_r(a)=0$ and
$f_1(a)=1$, i.e., we have $a\in \calZ(I_F + \langle f_1+1, f_r\rangle )$. 
Otherwise, we have $((f_r+1)f_1)(a)=0$. In this case~$a$ is a zero of
$(f_r+1)f_1$. Using Proposition~\ref{prop:preprocessing}, we deduce from
$a\in\calZ(I_F)$ that $a$ is a zero of $(f_i+1)f_j$ for all
$i,j\in\{1,\dots,r\}$ with $i<j$, as there is a path $f_i\to f_j$ in~$G$.

It follows that $(f_i+1)f_r\cdot f_1 + (f_r+1)f_1\cdot (f_i+1) = (f_i+1)f_1$ vanishes
at~$a$ for all $i\in\{1,\dots,r\}$. This shows that the point~$a$
is a zero of $(f_1+1)f_i + (f_i+1)f_1 = f_1+f_i$ for all $i\in\{2,\dots,r\}$. 
Finally, we get that $a\in\calZ(I_F + \langle f_1+f_i \mid i\in\{2,\dots,r\} \rangle )$, 
and the claim follows.
\end{proof}

This proposition allows us to introduce several simple decision heuristics.
In the next section, we will see that they prove quite effective on certain types of inputs.

\begin{remark}[Decision Heuristics]\label{rem:decheu}
Let $\sigma$ be a term ordering, let $F$ be a 2-XNF formula, and let $(L,V,E)$
be a $\sigma$-reduced acyclic IGS.

\textbf{MaxReach}. Find a source $f\in V$ such that the number of paths
      starting at $f\in V$ is maximal. Then we consider the decision
      $(D_f,\{f+1\})$. Since~$f$ is a source, the vertex~$f$ has no in-going edges. 
      Thus the skew-symmetry of~$G$ implies that $f+1$ has no out-going edges.
      This yields $D_{f+1}=\{f+1\}$.
      
\textbf{MaxBottleneck}. Instead of focusing on the first part of the
      decisions, another approach is to find $f\in V$ such that the sum of the
      number of paths ending in~$f$ and the number of paths starting at~$f$ is
      maximal. Then we consider the decision $(D_f, D_{f+1})$.
      
\textbf{MaxPath}. Let $f_1\to\dots\to f_r$ be a maximal path in~$(V,E)$. Then
      we consider the decision $(\{f_1+f_i \mid 1\leq i\leq r\},\, \{f_1+1,
      f_r\})$. Conceptually speaking, this means that instead of guessing vertices in
      the graph, we guess the edge $f_r\to f_1$, i.e., the polynomial $(f_r+1)f_1$. 
      In view of the proof of Proposition~\ref{prop:decheu} and of Remark~\ref{rem:motivation},
      this yields a strongly connected component of~$(V,E)$.
\end{remark}

While the first two of these heuristics are close to the classical approach to
decisions, the MaxPath heuristic is a rather new one. Note, however, that
these heuristics are just some initial suggestions and should be combined with
well-studied heuristics of established CDCL SAT solvers. Unfortunately, the adaptions 
of those heuristics to linerals are not straightforward.

The heuristics suggested in the previous remark are designed such that we can
compute them efficiently, i.e., in linear time and space. Let us give some more
information on how this can be done.

\begin{remark}[Efficient Implementation of Decision Heuristics]\label{rem:decheu_impl}
Recall that a topological ordering of a directed acyclic graph $(V,E)$ is a
linear ordering $\triangleleft$ of~$V$ such that $(f,g)\in E$ implies
$f\triangleleft g$, and that such an ordering can be computed in linear time and space (see~\cite{APT}).

\textbf{MaxReach}. For $f\in V$, denote the number of paths starting at~$f$
    by~$p_f$. Then we have $p_f = 1 + \sum_{(f,g)\in E}\, p_g$ for every $f\in V$. 
    This means that traversing the graph in a
    reverse topological order once allows us to find $p_f$ for all $f\in V$.
    In particular, the vertex $f\in V$ which maximizes $p_f$ can be found in linear time.
        
\textbf{MaxBottleneck}. Similarly, we can find the number of paths
    ending in each vertex $f\in V$ by a single traversal of the graph in
    topological order. Thus the vertex $f\in V$ which has the most paths
    starting and ending in~$f$ can be found by a total of two graph traversals.
        
\textbf{MaxPath}. For $f\in V$, denote the length of the longest path
    starting at~$f$ by~$\ell_f$. Then we have $\ell_f = 1+\max_{(f,g)\in E} \ell_g$.
    The value $\ell_f$ for all $f\in V$ can now be computed by a single
    traversal of the graph in a reverse topological order. By storing the vertex $g\in V$ 
    for which $\ell_g$ is largest at every $f\in V$ with $(f,g)\in E$, the
    path of length $\ell_f$ starting at~$f$ can be recovered in linear time.
    Altogether, the MaxPath heuristic can be implemented in linear time and space.
\end{remark}

%%%%%%%%%%%%%%%%%%%%%%%%%%%%%%%%%%%%%%%
%  Subsection 4.5
%%%%%%%%%%%%%%%%%%%%%%%%%%%%%%%%%%%%%%%

\subsection{Graph-based 2-XNF DPLL-Solving}

Finally, we have all the tools at our disposal to present our graph-based 2-XNF solver
which is based on the well-known DPLL-technique (see~\cite{DLL}).
\smallskip

\begin{algorithm}[ht]
  \DontPrintSemicolon
  \SetAlgoLongEnd
  \SetKwInOut{Input}{Input}
  \SetKwInOut{Output}{Output}
  \Input{An IGS $(L,V,E)$ for a formula~$F$ in 2-XNF, a term ordering~$\sigma$.}
  \Output{\UNSAT\ or an assignment $a\in\calS(F)$.}
  \BlankLine
  Let $(L,V,E)=\eGCP_\sigma(L,V,E)$. \tcp*{propagation}
  Let $L_{\mathtt{FL}} = \FLS(L,V,E)$ and adjoin $L_{\mathtt{FL}}$ to~$L$. \tcp*{in-processing}
  \lIf{$L_{\mathtt{FL}}\neq\emptyset$}{go to Line~$1$.}
  \lIf{$1\in\langle L \rangleF$}{\Return\ \UNSAT}
  \lIf{$E=\emptyset$}{\Return $a\in\calZ(L) \,\subseteq\, \FF_2^n$.}
  Use Remark~\ref{rem:decheu} to compute a decision $(L_0,L_1)$ for~$(L,V,E)$. \tcp*{decision}
  \lIf{$\SOLVER_{\sigma}(L\cup L_0,V,E)$ returns $a\in\FF_2^n$}{\Return\ $a$.}
  \lElse{\Return $\SOLVER_{\sigma}(L\cup L_1,V,E)$.}
  \caption{\SOLVER \,--\, Graph-Based 2-XNF DPLL-Solver}
\end{algorithm}

\begin{proposition}\label{prop:alg:solver}
Let~$F$ be a formula in 2-XNF with an implication graph structure $(L,V,E)$, and let~$\sigma$ be a term ordering. Then $\SOLVER$
is an algorithm which returns $\UNSAT$ if and only if $\calS(F)=\emptyset$.
Otherwise, it returns an element $a\in\calS(F)$.
\end{proposition}
\begin{proof}
First notice that Line~1 ensures that the IGS $(L,V,E)$ is always $\sigma$-reduced and acyclic. Hence Line~6 can be performed efficiently, as explained in Remark~\ref{rem:decheu_impl}.

Next we show the finiteness of the procedure.
In every iteration of Lines~$1$-$3$ where $L_{\mathtt{FL}}\neq\emptyset$ the IGS $(L,V,E)$ decreases strictly
w.r.t.~$\prec$. By Lemma~\ref{lem:DCC}, we eventually reach $L_{\mathtt{FL}}=\emptyset$ in Line~$3$,  and the loop stops after finitely many steps.
For the finiteness of the recursive calls observe that if $(L_0,L_1)$ is a decision for $(L,V,E)$ as in Line~6, then $\dim_{\FF_2} \langle L \rangleF < \dim_{\FF_2} \langle L\cup L_0 \rangleF$ and $\dim_{\FF_2} \langle L \rangleF < \dim_{\FF_2} \langle L\cup L_1 \rangleF$.
This means that the dimension of $\langle L\rangleF$ increases strictly with every recursive call. Now it suffices to note that this dimension is bounded from above by $n+1$, and in case $\dim_{\FF_2} \langle L\rangleF = n+1$ we have $1\in\langle L\rangleF$, i.e., the procedure terminates in Line~4.

To prove correctness,
note that if the algorithm terminates in Line~4, then $F$ cannot have any solution since $\langle L\rangleF \subseteq I_F$. Similarly, if it terminates in Line~5, the implication graph must be empty and we get $I_F=\langle L\rangle$, i.e., $a\in\calZ(L)=\calZ(I_F) = \calS(F)$.
Next we show by induction on~$d$ that the output in all lines is correct if $\dim\langle L\rangleF = d$ for $d\in\{0,\dots,{n+1}\}$.
Note that $\dim_{\FF_2}\langle L\rangleF = n+1$ implies $1\in\langle L\rangleF \subseteq I_F$, i.e., the algorithm terminates already in Line~4 and is correct by the above.
Now suppose that the algorithm terminates correctly if $\dim_{\FF_2}\langle L\rangleF > s$ for some $s\in\{0,\dots,n\}$ and let $\dim_{\FF_2} \langle L\rangleF {= s}$. 
It suffices to consider the case where the algorithm terminates in Lines~7 or~8. Note that by definition of the decision $(L_0,L_1)$ from Line~6 we have $\calZ(I_F)\subseteq \calZ(I_F +\langle L_0\rangle ) \cup \calZ(I_F +\langle{L_1}\rangle)$, and as above the dimension of $\langle L\rangleF$ is strictly smaller than the dimensions of $\langle L\cup L_0\rangleF$ and $\langle L\cup L_1\rangleF$, respectively. 
Thus the recursive call in Line~7 terminates correctly, i.e., returns \UNSAT if and only if $\calZ(I_F+\langle{L_0}\rangle) =\emptyset$, otherwise it returns $a\in\calZ(I_F+\langle{L_0}\rangle) \supseteq\calZ(I_F)$. 
If the algorithm does not terminate here, then we must have $\calZ(I_F+\langle L_0\rangle) = \emptyset$ and the algorithm terminates with the recursive call in Line~8. Analogous to the call in Line~7, we get \UNSAT if and only if $\calZ(I_F+\langle L_1\rangle)=\emptyset$, which occurs if and only if $\calZ(I_F)\subseteq \calZ(I_F +\langle L_0\rangle ) \cup \calZ(I_F +\langle{L_1}\rangle) = \emptyset$. Otherwise it returns a satisfying assignment $a\in\calZ(I_F+\langle L_1\rangle) \,\supseteq\, \calZ(I_F) = \calS(F)$ of~$F$.
\end{proof}

%%%%%%%%%%%%%%%%%%

To obtain an efficient implementation we need appropriate data structures which
support fast backtracking. The following method allows us to avoid creating a copy of the entire
implication graph structures in the recursive calls of Lines~$7$ and~$8$.

\begin{remark}[Data Structures for Implication Graph Structures]\label{rem:DataStructures}
In order to store an implication graph structure $(L,V,E)$ internally, it is
beneficial to actually store a graph $(V',E')$ based on integer vertices $V'\subseteq\ZZ$ 
and a map $\lambda:\; V'\to \LL_n$ such that $V=\lambda(V')$ and such that
$E = \{\,(\lambda(v),\lambda(w)) \mid (v,w)\in E'\}$. 

Let us suggest two data structures, one for the labeling map~$\lambda$, 
and one for the graph $(V',E')$ which are tailored towards efficient backtracking.
\begin{enumerate}
\item[(a)] To efficiently represent~$\lambda$, we use a \textit{prefix tree}, 
i.e., a tree whose non-root vertices are elements of $\{1,x_1,\dots,x_n\}$, where the 
children of every node are bigger than their parent w.r.t.~a term ordering $\sigma$, 
and where the root is $t_0=0$.
Then every vertex $v\in V'$ is associated to a vertex $\eta(v)$ of
the tree such that the unique path starting at the root 
$t_0\to \dots \to t_r=\eta(v)$ satisfies $\lambda(v) = t_0+\dots+t_r$.

Note that insertion can be performed in amortized linear time in the size of 
$\Supp(\lambda(v))$ if the children are accessed by hash maps, and deletion can be performed in constant time.
If~$\lambda$ needs to be copied, it suffices to copy $\eta(v)$ for every $v\in V'$. 
The actual linear polynomials $\lambda(v)$ are not copied. For the backtracking, 
we simply replace $\eta$ internally, and the previous $\lambda$ is restored immediately.

\item[(b)] For the graph itself, we suggest to use a modified \textit{lean
hybrid graph representation}, as devised in \cite{AJM,AKPWS,ALN}.
This data structure was proposed only for undirected graphs, but an
extension to directed skew-symmetric graphs is possible. The data
structure is rather advanced and allows backtracking of edge deletions
and vertex contractions in constant time. In particular, it allows us to
store any state of the graph with a space complexity of $\calO(\#V)$.
Backtracking to such a previous state has complexity $\calO(\#V)$.
\end{enumerate}
Altogether, it is possible to implement the algorithm with a space complexity
of $\calO((n+1)\cdot \#V + \#E)$, where $(V,E)$ is part of the initial trivial
IGS.
\end{remark}

Notice that \SOLVER\ is based on the well-known DPLL framework. An extension to a
conflict-driven clause learning (CDCL) directive encounters the following obstacles. 

\begin{remark}[Conflict-Driven XNF Clause Learning]
Although the 2-XNF theory originates from the~\SRES\ proof system which in turn is
a generalization of classical resolution, it is not easy to extend conflict-driven 
clause learning to 2-XNF instances. This is mainly due to two problems: 
\begin{enumerate}
\item[(1)] The resolvent of two clauses may be the \textit{zero clause}, i.e., 
resolving the conflict clause may lead to a clause that is trivially satisfied 
(see~\cite{Hor}).

\item[(2)] In general, the resolvent is not in 2-XNF, i.e., it cannot be added to 
the implication graph structure in a straightforward way.
\end{enumerate}
Overcoming these obstacles is an important objective of future research, because
CDCL techniques promise significant speed-ups of XNF solvers. 
\end{remark}

%%%%%%%%%%%%%%%%%%%%%%%%%%%%%%%%%%%%%%%%%%%%%%%%%%%%%%%%%%%%%%%%
%
%  Section 5: Experiments and Timings
%
%%%%%%%%%%%%%%%%%%%%%%%%%%%%%%%%%%%%%%%%%%%%%%%%%%%%%%%%%%%%%%%%

\section{Experiments and Timings}
\label{sec:timings}

In this section we evaluate the methods of Section~\ref{sec:solving} on random
2-XNF instances and on instances coming from round-reduced \Asconnb\ key-recovery attacks.
For comparison, we ran our \cpp implementation of Algorithm \SOLVER,
which we named \txnfsolver, against SAT solvers with XOR support, i.e., CNF-based SAT solvers
that can read and process XOR constraints on the variables natively. We say that
formulas of the type processed by these solvers are in CNF-XOR. 

State-of-the-art SAT solvers that support CNF-XOR input are~\cms
(see~\cite{SNC}), an established CDCL-based solver, and~\xnfsat
(see~\cite{NLFHB}), which is based on a \textit{stochastic local search}
approach, i.e., it can only be used on satisfiable instances. (Note that \xnfsat,
despite its name, cannot work with XNFs as introduced in this article. It only
supports CNF-XOR instances.)
To use these solvers on XNF instances, we use the following reduction.

\begin{remark}\label{rmk:XNFtoXCNF}
Let $F$ be a 2-XNF formula involving $n$ variables.
Then we can write the XNF clauses of~$F$ as $C_1,\dots,C_r$, $L_1,\dots,L_s$,
where $C_i = L_{i,1} \lor L_{i,2}$ with linerals $L_{i,j}$, and where
$L_1,\dots,L_s$ are already linerals.
Now we introduce $2r$ additional variables $Y_{i,j}$ and consider the CNF-XOR formula~$G$ 
consisting of the clauses $C'_i = Y_{i,1} \lor Y_{i,2}$, the XOR
constraints $\neg Y_{i,j}\oplus L_{i,j}$ for $i\in\{1,\dots,r\}$ and
$j\in\{1,2\}$, and the original XOR constraints $L_1,\dots,L_s$. Then we
have $\calZ(F) \equiv_n \calZ(G)$.
\end{remark}

Furthermore, a 2-XNF instance can also be seen as a system of quadratic Boolean
polynomial equations that can be solved by an algebraic solver
such as \PolyBoRi (see~\cite{BD}). This package offers an implementation of the
Buchberger algorithm adapted to Boolean polynomial rings and employs highly
optimized data structures. Additionally, we consider the solver \Bosphorus (see~\cite{CSCM}),
which employs both algebraic and logical reasoning, and processes ANF (and CNF) input.
For instances with fewer than $40$ variables, we also compare the solvers to
\xnfbf, our \cpp implementation of a \textit{brute-force} XNF solver.
Finally, we also consider the winner of the 2023 SAT competition~\sbva (see~\cite{HGH}) which processes CNF inputs.
The CNF files were generated from the CNF-XOR representation by converting the additional XOR constraints on the variables to a set of CNF clauses. Since a direct encoding of long XORs results in exponentially many CNF clauses, they are split using new variables such that we only consider direct encodings of XOR constraints involving at most~$5$ variables.
This corresponds to a linear encoding with cutting number~$5$ (see~\cite{NLFHB}).
\smallskip

All experiments were run on an \textit{Intel Xeon E5-2623 v3} processor with
\textit{128}GB of RAM under Debian \textit{10}. We used \cms version~$5.8$,
\xnfsat version~$03v$, \Bosphorus version~$3.0$, and~\sbva with~\cadical~2.0 (see~\cite{BFFH}).

\subsection*{Random 2-XNF Clauses}

First we consider \textit{random} 2-XNF instances involving $n$ variables and $m$
clauses. Every clause in the formula is generated by picking two linerals
uniformly at random in $\LL_n\setminus\FF_2$.
With $m=3\cdot n$ and $n\in\{21,\dots,40\}$, experiments showed that such an
instance is \UNSAT with a probability of at least~$98\%$. If a solution is
desired, we simply choose $a\in\FF_2^n$ at random and for every clause that is
not satisfied by~$a$, we randomly flip the constant of one of the two linerals. 
This ensures that $a$ indeed forms a satisfying assignment of the
generated 2-XNF instance.
Two random benchmark suites are considered, each containing $400$ random
instances with $n\in\{21,\dots,40\}$ variables in $m=3\cdot n$ clauses, where we
have $20$ instances for every~$n$. One set contains only satisfiable instances,
the other only unsatisfiable ones.

\begin{figure}
    \centering
    \begin{subfigure}[b]{0.85\textwidth}
        \centering
        \includegraphics[width=\textwidth]{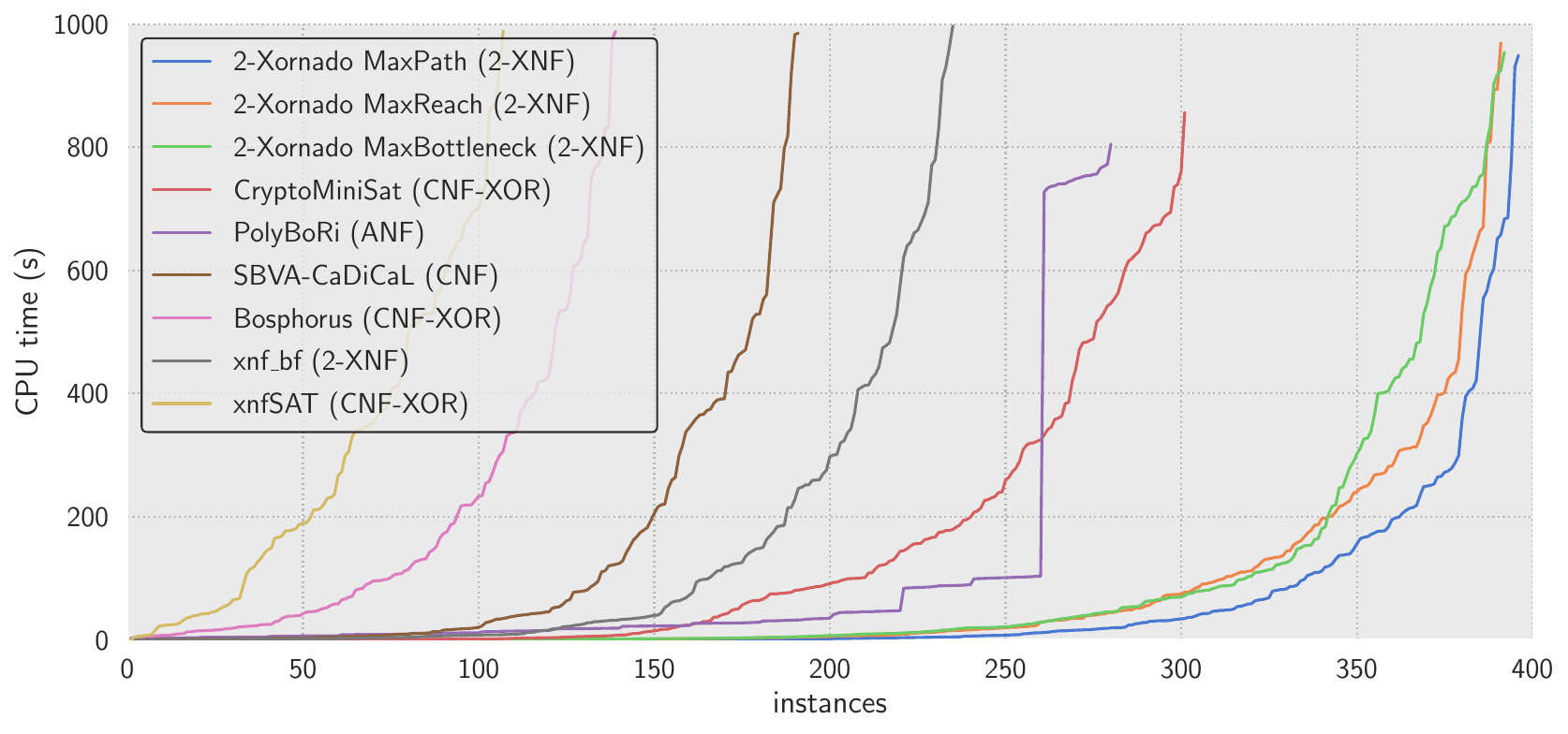}
        \caption{
          Benchmark suite consisting of $400$~random \textit{satisfiable} 2-XNF
          instances in $n$ indeterminates and $3n$ clauses
          where~$n\in\{21,\dots,40\}$.
        }
        \label{fig:cactus_sat}
    \end{subfigure}
    \vfill
    \begin{subfigure}[b]{0.85\textwidth}
        \centering
        \includegraphics[width=\textwidth]{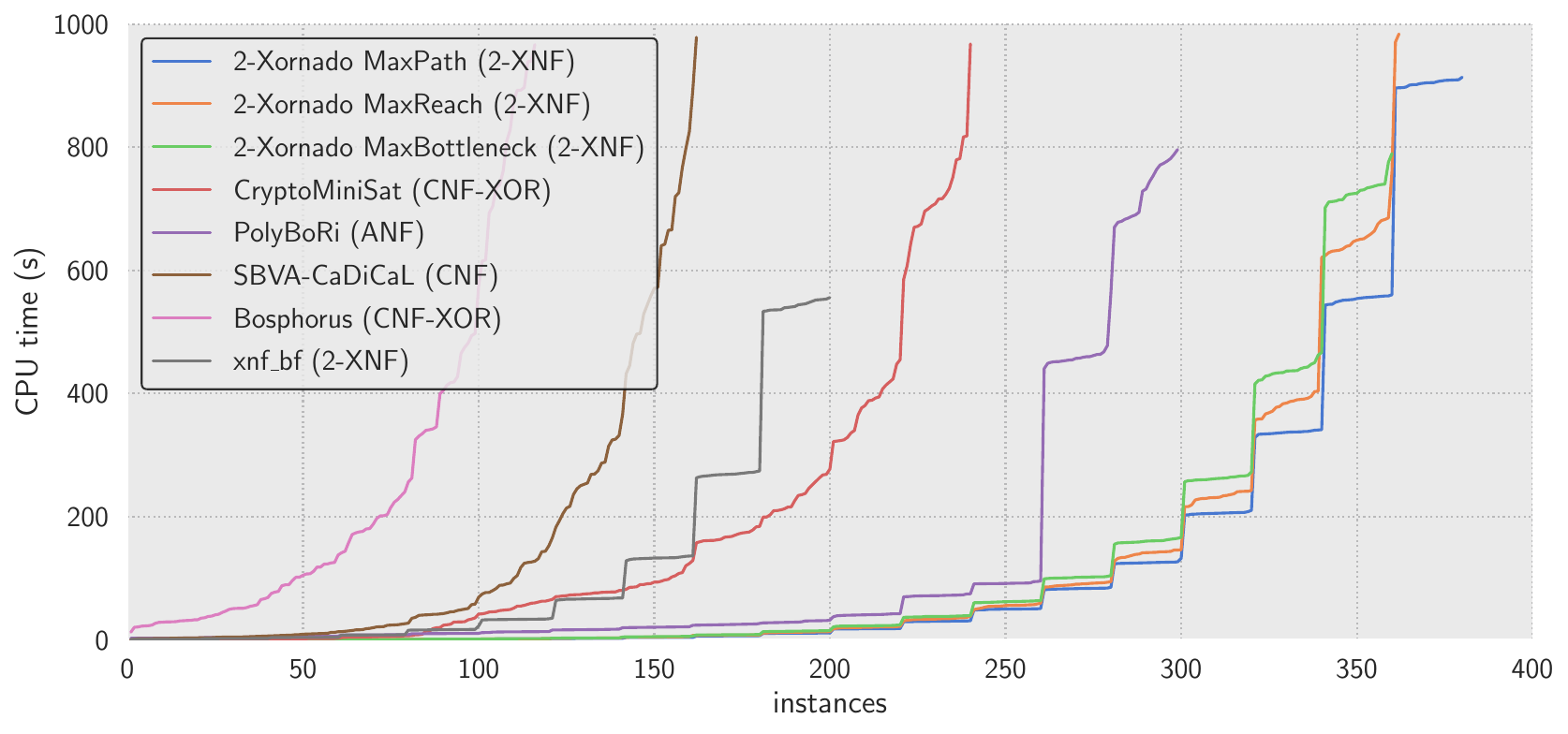}
        \caption{
          Benchmark suite consisting of $400$~random \textit{unsatisfiable}
          2-XNF instances in $n$ indeterminates and $3n$ clauses
          where~$n\in\{21,\dots,40\}$.
        }
        \label{fig:cactus_unsat}
    \end{subfigure}
    \caption{
      Cactus plots for the random benchmark suites.
      %comparing our solver
      %\txnfsolver equipped with different decision heuristics to a brute-force
      %approach \xnfbf, the CNF-XOR SAT solvers~\xnfsat, \cms, an algebraic
      %solver based on \PolyBoRi, and~\Bosphorus.
    }
    \label{fig:bench_rand}
\end{figure}

The cactus plots in Figure~\ref{fig:bench_rand} show that such small random
instances are hard for state-of-the-art CNF and CNF-XOR solvers.
In particular, we see that~\xnfsat and~\sbva are even out-performed by a simple brute-force implementation.
Algebraic solving with \PolyBoRi performs not significantly worse than \cms.
While the plot clearly shows that \txnfsolver performs best on this benchmark,
one should note that this is not due to some clever data structures that allow
very fast propagation. The main reason for its better overall performance is that
the required number of decisions of \txnfsolver (with any heuristic) is smaller
by a factor of $60$-$80$ compared to the number of decisions taken by~\cms.

\subsection*{Round-Reduced Ascon Key Recovery Attacks}

Our second benchmark set consists of instances related to key-recovery attacks
on round-reduced versions of the cipher \Asconnb\
(see~\cite{DEMS}). In particular, we consider attacks where the
$128$-bit nonce and the $320$-bit internal state are known and the goal is to
\textit{undo} the initialization step consisting of $12$~rounds in order to obtain the
$128$-bit secret key. If this problem can be solved efficiently, the cipher is
broken in a nonce-misuse scenario, see~\cite{BCP}. Here we
consider round-reduced variants: $20$~instances with $2$ rounds, $20$~instances
with $3$ rounds and knowledge of the first~$k$ key bits for each
$k\in\{55,\dots,64\}$, and $20$~instances with $4$ rounds and knowledge of the
first~$k$ key bits for each $k\in\{92,\dots,100\}$. The instances were generated
by applying \QAnfToTXnf{} to a polynomial representation of the cryptosystem,
see Example~\ref{exa:ascon_sbox}, augmented with some additional XNF clauses, which
speed up propagation in~\GCP.

\begin{figure}
    \centering
    \includegraphics[width=0.85\textwidth]{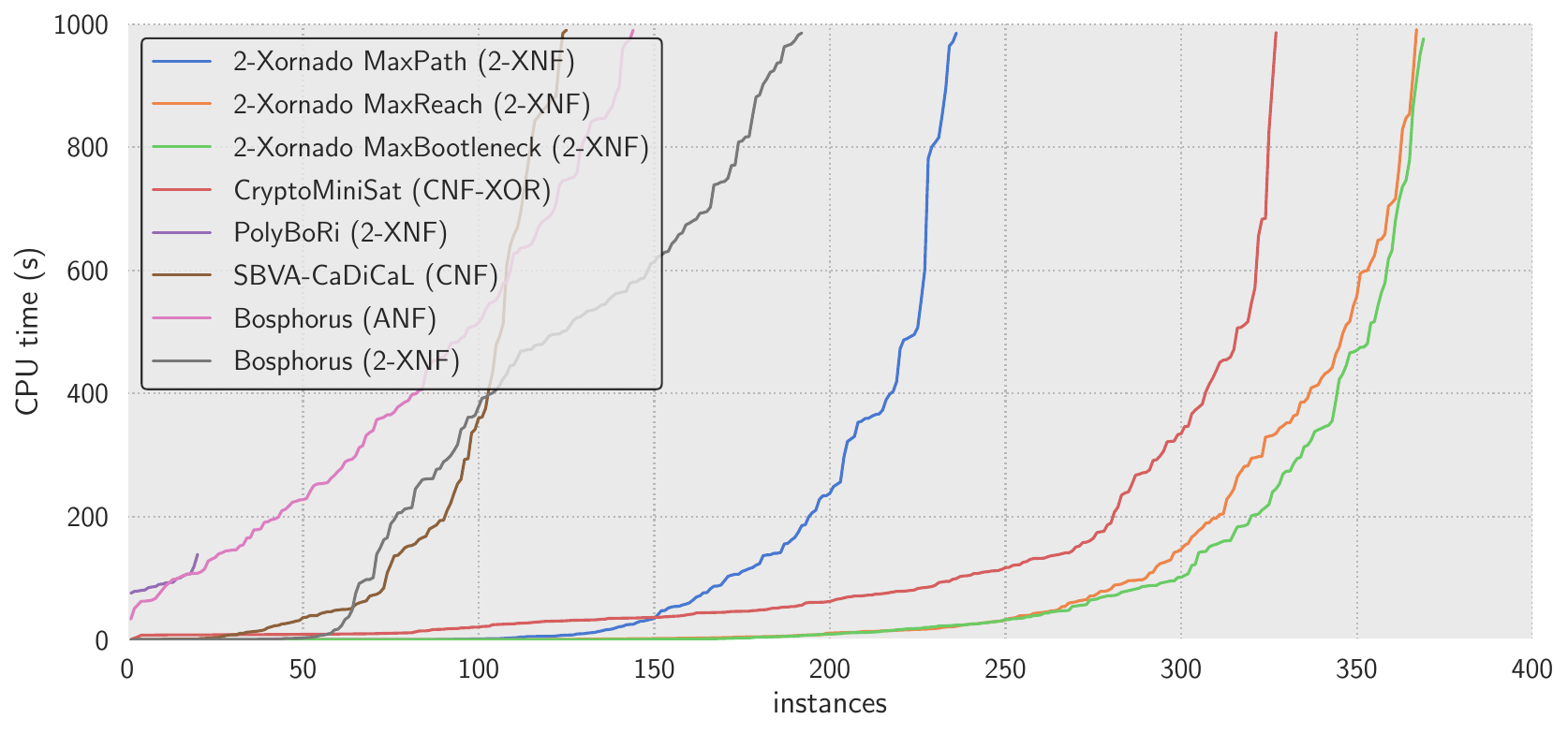}
    \caption{
      Cactus plot for the benchmark suite consisting of $400$~satisfiable
      instances related to key-recovery attacks on round-reduced \Asconnb.
    }
    \label{fig:bench_rr_ascon}
\end{figure}

Figure~\ref{fig:bench_rr_ascon} contains a cactus plot for these cryptographic
instances. Here \xnfsat was not included due to its bad performance on the random
set. It turns out that for the 2-round version \txnfsolver can already solve all 
instances (starting with the trivial IGS) during pre-processing in less than $0.3$ seconds on average.
On these instances, \cms requires more than $80\,000$ decisions and several seconds; and \sbva about $20\,000$ decisions and one second.
Our solver \txnfsolver with the MaxBottleneck or the MaxReach heuristic and in-processing
with \FLS also performs very well on the remaining benchmark and comes out as the average best solver.
The bad performance with the MaxPath heuristic may be attributed to the fact 
that the corresponding decision linerals contain more variables and therefore \eGCP execution 
requires more time, increases the average length of linerals of the implication graph vertices, 
i.e., increases its memory footprint, and thereby makes backtracking more expensive.
It should be noted that on the 4-round instances the CNF-XOR solver \cms 
had a better performance with fewer timeouts.
So its advanced decision heuristics, the highly optimized data structures, and the conflict-learning methods 
do pay off on larger instances. 
Nonetheless \txnfsolver still requires fewer decisions by a factor of $60$-$80$.
The CNF-SAT solver~\sbva, however, with data structures and conflict-learning methods similar to~\cms could only solve $8$ of these instances.
This highlights the effectiveness of the encodings of these innately XOR-rich problems in CNF-XOR and XNF.
Also note that \PolyBoRi could not solve a single instance when given in its ANF format,
however when feeding it with the system of quadratic equations corresponding to the 2-XNF,
some instances could be solved. The situation for~\Bosphorus is similar, with better 
performance on the input that comes from our XNF encoding.

\section*{Conclusions}

A generalization of the well-known CNF that
allows compact representations of XOR-rich problems, like those originating from
cryptographic attacks, has been introduced. On top of that we generalized pre- and in-processing
techniques and introduced a DPLL-based solving algorithm with a simplistic decision
heuristic that outperforms other state-of-the-art solving approaches on random instances and
on problems originating from cryptographic attacks on \Asconnb. An extension to CDCL-based
solving is in preparation and better decision heuristics will be investigated.

\medskip\noindent
{\bf Acknowledgements.} During part of this research, the second author was supported by
the DFG project \textit{Algebraische Fehlerangriffe} KR 1907/6-2. The first author
gratefully acknowledges \textit{Cusanuswerk e.V.}~for financial support.
%\thanks{During part of this research, the second author was supported by
%the DFG project \textit{Algebraische Fehlerangriffe} KR 1907/6-2. The first author
%gratefully acknowledges \textit{Cusanuswerk e.V.}~for financial support.}

%%%%%%%%%%%%%%%%%%%%%%%%%%%%%%%%%%%%
%  The Bibliography
%%%%%%%%%%%%%%%%%%%%%%%%%%%%%%%%%%%%


\begin{thebibliography}{10}


\bibitem{AJM}
F.\ N.\ Abu-Khzam, K.\ A.\ Jahed, and A.\ E.\ Mouawad,
A hybrid graph representation for exact graph algorithms, preprint 2014,
available at \url{arxiv.org/pdf/1404.6399.pdf} (accessed on 23 February 2023).


\bibitem{AKPWS}
F.\ N.\ Abu-Khzam, D.\ Kim, M.\ Perry, K.\ Wang, and P.\ Shaw,
Accelerating vertex cover optimization on a {GPU} architecture,
in: \emph{Int. Symposium on Cluster, Cloud and Grid Computing (CCGRID)}, Washington 2018, 
IEEE  Xplore, pp.616--625.


\bibitem{ALN}
F.\ N.\ Abu-Khzam, M.\ A.\ Langston, and C.\ P.\ Nolan,
A hybrid graph representation for recursive backtracking algorithms,
in: \emph{Frontiers in Algorithmics (FAW 2010)}, LNCS 6213, 
Springer-Verlag, Berlin 2010, pp. 136--147.


\bibitem{ACGKLRS}
M.\ Albrecht, C.\ Cid, L.\ Grassi, D.\ Khovratovich, R.\ L{\"u}ftenegger, C.\ Rechberger, and M.\ Schofnegger,
Algebraic cryptanalysis of {STARK}-friendly designs: Application to {MARVELlous} and {MiMC},
in: \emph{Proc. Advances in Cryptology (ASIACRYPT 2019)}, Kobe 2019, LNCS 11923,  
Springer Int. Publ., Cham 2019, pp. 371--397.


\bibitem{APT} B.\ Aspvall, M.\ F.\ Plass, and R.\ E.\ Tarjan,
A linear-time algorithm for testing the truth of certain quantified boolean formulas,
Inform.\ Process.\ Lett.\ {\bf 8} (1979), 121--123.


\bibitem{BCP} J.\ Baudrin, A.\ Canteaut, and L.\ Perrin, Practical cube attack against nonce-misused {A}scon,
IACR Transactions on Symmetric Cryptology \textbf{4} (2022), 120--144.


\bibitem{BFFH} A. Biere, T. Faller, K. Fazekas, M. Fleury, N. Froleyks, and F. Pollitt, {CaDiCaL 2.0}, in: \emph{Proc. Computer Aided Verification (CAV 2024)}, Montreal 2024, LNCS 14681, Springer Nature Switzerland, Cham 2024, pp. 133-152.


\bibitem{Bri} M.\ Brickenstein, \textit{Boolean Gr\"obner Bases: Theory, Algorithms 
and Applications}, Springer-Verlag, Berlin 2010.


\bibitem{BD} M.\ Brickenstein and A.\ Dreyer,
{PolyBoRi}: A framework for {G}r{\"o}bner-basis computations with {B}oolean polynomials,
J.\ Symbolic Comput.\ {\bf 44} (2009), 1326--1345.


\bibitem{CD} W.\ Castryck and T.\ Decru, An efficient key recovery attack on {SIDH},
in: \emph{Proc. Advances in Cryptology (EUROCRYPT 2023)}; Lyon 2023, LNCS 14008, 
Springer Int. Publ., Cham 2023, pp. 423--447.


\bibitem{CSCM} D.\ Choo, M.\ Soos, K.\ M.\ A.\ Chai, and K.\ S.\ Meel,
Bosphorus: Bridging ANF and CNF solvers, in:
\emph{Proc. Design, Automation, and Test in Europe (DATE)}, Florence 2019, 
IEEE Xplore, pp.\ 468-473.


\bibitem{CKPS}
N.\ Courtois, A.\ Klimov, J.\ Patarin, and A.\ Shamir,
Efficient algorithms for solving overdefined systems of multivariate polynomial equations,
in: \emph{Proc. Advances in Cryptology (EUROCRYPT 2000)}, Brugge 2000, LNCS 1807, 
Springer-Verlag, Berlin 2000, pp.\ 392--407.


\bibitem{CSSV}
N.\ Courtois, P.\ Sepehrdad, P.\ Su{\v{s}}il, and S.\ Vaudenay, The {E}lim{L}in algorithm revisited,
in: \emph{Proc. Fast Software Encryption (FSE 2012)}, Washington 2012, LNCS 7549, 
Springer-Verlag, Berlin 2012, pp.\ 306--325.


\bibitem{DK} J.\ Danner and M.\ Kreuzer, A fault attack on {KC}ipher-2,
Int.\ J.\ Comput.\ Math.\ Comput.\ Syst.\ Theory {\bf 6} (2021), 281--312.


\bibitem{DSHC} S.\ K.\ Dash, S.-B.\ Scholz, S.\ Herhut, and B.\ Christianson,
A scalable approach to computing representative lowest common ancestor in directed acyclic graphs, 
Theoret.\ Comput.\ Sci.\ {\bf 513} (2013), 25--37.


\bibitem{Dav} J.\ Davies, Solving MAXSAT by decoupling optimization and satisfaction, 
dissertation, University of Toronto, Toronto 2013.


\bibitem{DLL} M.\ Davis, G.\ Logemann, and D.\ Loveland,
A machine program for theorem proving, Commun. ACM {\bf 5} (1962), 394--397.


\bibitem{DEMS} C.\ Dobraunig, M.\ Eichlseder, F.\ Mendel, and M.\ Schl{\"a}ffer,
\textit{Ascon v1.2}: Technical report, National Institute of Standards and Technology, 2019.

\bibitem{DMV}  J.\ M.\ Dudek, K.\ S.\ Meel, and M.\ Y.\ Vardi, The hard problems are almost everywhere for random CNF-XOR formulas. in: \emph{Proc. Int. Joint Conference on Artificial Intelligence (IJCAI'17)}, Melbourne, 2017; pp.~600--606.


\bibitem{DKMNPW}
A.\ D.\ Dwivedi, M.\ Klou{\v{c}}ek, P.\ Morawiecki, I.\ Nikoli{\'c}, J.\ Pieprzyk, and S.\ W{\'o}jtowicz,
SAT-based cryptanalysis of authenticated ciphers from the CAESAR competition,
in: \emph{Proc. Int. Joint Conference on e-Business and Telecommunications (ICETE 2017)},
SECRYPT, Madrid, 2017; pp. 237--246.


\bibitem{EKMS} G.\ Emdin, A.\ S.\ Kulikov, I.\ Mihajlin, and N.\ Slezkin, CNF Encodings of Symmetric Functions, Theory Comput.\ Sys.\ (2024).


\bibitem{HGH} A.\ Haberlandt, H.\ Green, and M.\ J.\ H.\ Heule, Effective Auxiliary Variables via Structured Reencoding, in: \emph{Proc. Theory and Applications of Satisfiability Testing (SAT 2023)}, Alghero 2023, LIPIcs~271, Leibniz-Zentrum f{\"u}r Informatik, Dagstuhl~2023, pp.~11:1-11:19.


\bibitem{HE} H.\ Hadipour and M.\ Eichlseder, Autoguess: a tool for finding guess-and-determine 
attacks and key bridges, in: \emph{Proc. Applied Cryptography and Network Security (ACNS 2022)}, Rome 2022, 
LNCS 13269, Springer Nature Switzerland, Cham 2022, pp.~230--250. 


\bibitem{HJ} C.\ S.\ Han and J-H.\ R.\ Jiang, When {B}oolean satisfiability meets {G}aussian 
elimination in a simplex way, in: \emph{Proc. Computer Aided Verification (CAV 2012)}, Berkeley 2012, 
LNCS 7358, Springer-Verlag, Berlin 2012, pp.\ 410--426.


\bibitem{HMB} M.\ J.\ H.\ Heule, J.\ Matti, and A.\ Biere, Revisiting hyper binary resolution,
in: \emph{Integration of AI and OR Techniques in Constraint Programming for Combinatorial 
Optimization Problems (CPAIOR 2013)}, LNCS 7874, Springer-Verlag, Berlin 2013, pp. 77--93.


\bibitem{Hor} J.\ Hor\'a\v cek, Algebraic and logic solving methods for 
cryptanalysis, dissertation, Universit\"at Passau, Passau 2020.


\bibitem{HK1} J.\ Hor\'a\v cek and M.\ Kreuzer, Refutation of products of linear 
polynomials, in: Proc.\ Third Int. Workshop on Satisfiability Checking and Symbolic 
Computation (SC\^{\,}2), Oxford 2018, available at \url{http://ceur-ws.org/Vol-2189/}.


\bibitem{HK2} J.\ Hor\'a\v cek and M.\ Kreuzer, On conversions from CNF to ANF, 
J.\ Symbolic Comput.\ {\bf 100} (2020), 164--186.


\bibitem{JK} P.\ Jovanovic and M.\ Kreuzer, Algebraic attacks using SAT-solvers,
Groups Complexity Cryptology \textbf{2} (2010), 247--259.


\bibitem{KR1} M. Kreuzer and L. Robbiano, {\it Computational Commutative Algebra 1},
Springer-Verlag, Berlin 2000.


\bibitem{LNV} F.\ Lafitte, J.\ Nakahara, and D.\ Van Heule, Applications of {SAT} 
solvers in cryptanalysis: finding weak keys and preimages,
J. Satisf.\ Boolean Model.\ Comput.\ {\bf 9} (2014), 1--25.

\bibitem{LJN} T.\ Laitinen, T.\ Junttila, and I.\ Niemelä, Conflict-Driven XOR-Clause Learning. in: \emph{Proc. Theory and Applications of Satisfiability Testing (SAT 2012)}, Trento 2012, LNCS~7317, Springer-Verlag, Berlin~2012, pp.~383–396.


\bibitem{LZLW} A.\ Leventi-Peetz, O.\ Zendel, W.\ Lennartz, and K.\ Weber,
{CryptoMiniSat} switches-optimization for solving cryptographic instances,
in: \emph{Proc. Pragmatics of SAT 2015 and 2018}, EPiC Series in Computing {\bf 59},
EasyChair 2019, pp. 79-93.


\bibitem{MZ} I.\ Mironov and L.\ Zhang, Applications of {SAT\ solvers to cryptanalysis of hash functions,
in: \emph{Proc. Theory and Applications of Satisfiability Testing (SAT 2006)}, Seattle 2006, LNCS 4121, 
Springer-Verlag, Berlin 2006, pp.\ 102-115.


\bibitem{MMZZM} M.\ W.\ Moskewicz, C.\ F.\ Madigan, Y.\ Zhao, L.\ Zhang, and S.\ Malik,
Chaff: engineering an efficient {SAT}} solver, in: \emph{Proc. Design Automation 
Conference (DAC), Las Vegas 2001}, ACM, New York 2001, pp.\ 530-535.


\bibitem{NLFHB}
W.\ Nawrocki, Z.\ Liu, A.\ Fr{\"o}hlich, M.\ J.\ H.\ Heule, and A.\ Biere,
XOR local search for Boolean brent equations,
in: \emph{Theory and Applications of Satisfiability Testing (SAT 2021)}, LNCS 12831,  
Springer Nature Switzerland, Cham 2021, pp.\ 417--435.



\bibitem{ST} R.\ Sebastiani and P.\ Trentin,
{OptiMathSAT}: a tool for optimization modulo theories,
J.\ Automat.\ Reason.\ \textbf{64} (2020), 423--460.


\bibitem{SM} M.\ Soos and K.\ S.\ Meel, {BIRD}: Engineering an efficient {CNF}-{XOR} {SAT} solver
and its applications to approximate model counting, in: \emph{Proc. AIII Conference on Artificial 
Intelligence 2019}, vol.\ 33, AIII Press, Palo Alto 2019, pp.\ 1592--1599.


\bibitem{SNC} M.\ Soos, K.\ Nohl, and C.\ Castelluccia, Extending {SAT} solvers to 
cryptographic problems, in: \emph{Theory and Applications of Satisfiability Testing (SAT 2009)}, 
LNCS 5584, Springer-Verlag, Berlin 2009, pp. 244--257.


\bibitem{Tar} R.\ Tarjan, Depth-first search and linear graph algorithms,
SIAM J.\ Comput.\ {\bf 1} (1972), 146--160.


\bibitem{TID} M.\ Trimoska, S.\ Ionica, and G.\ Dequen, Parity {(XOR)} reasoning for the 
index calculus attack, in: \emph{Proc. Principles and Practice of Constraint Programming (CP 2020)}, 
Louvain-la-Neuve 2020, Springer Int. Publ., Cham 2020, pp. 774-790.


%\bibitem{brain2017benchmarking}
%M.\ Brain, J.\ H.\ Davenport, A.\ Griggio,
%\newblock Benchmarking Solvers, SAT-style.
%\newblock In: \emph{Proc. Satisfiability Checking and Symbolic Computation (SC\^{\,}2)}, Kaiserslautern 2017, 
%available at \url{http://ceur-ws.org/Vol-1974/}.

%\bibitem{apcocoa}
%The ApCoCoA Team, {ApCoCoA}: Applied Computations in Computer Algebra, 
%\url{https://apcocoa.uni-passau.de}, 2023


\end{thebibliography}
\end{document}